\documentclass[english, 12pt, renew]{article}
\usepackage{tikz}
\usepackage{verbatim}
\usetikzlibrary{decorations.pathreplacing,arrows.meta,patterns,shapes, calligraphy}
\usepackage{mathtools}
\pdfoutput=1
\usepackage{ae,aecompl}
\usepackage[T1]{fontenc}
\usepackage[utf8]{inputenc}
\usepackage[margin = 1.in]{geometry}
\usepackage{fancyhdr}
\usepackage{color}

\usepackage{bm}

\usepackage{babel}
\usepackage{amsmath}

\usepackage{amsthm}
\usepackage{amssymb}
\usepackage[authoryear]{natbib}
\usepackage{blindtext}
\usepackage{graphicx}
\usepackage{subcaption}
\usepackage[utf8]{inputenc}
\usepackage[export]{adjustbox}
\usepackage{wrapfig}
\usepackage{hyperref}
\usepackage{breakurl}
\usepackage{tcolorbox}
\usepackage{breakurl}
\usepackage[normalem]{ulem}
  \theoremstyle{definition}
  \newtheorem{definition}{\protect\definitionname}
  \theoremstyle{plain}
  
\theoremstyle{plain}
\newtheorem{theorem}{\protect\theoremname}
  \theoremstyle{plain}
  \newtheorem{lemma}{\protect\lemmaname}
  \theoremstyle{remark}
  
  \theoremstyle{plain}
  \newtheorem{cor}{\protect\corollaryname}
 \theoremstyle{definition}
  \newtheorem{example}{\protect\examplename}
   \theoremstyle{definition}
  
  \theoremstyle{plain}
  
  \theoremstyle{plain}
  \newtheorem{proposition}{\protect\propositionname}
  \theoremstyle{definition}

  \theoremstyle{remark}

\usepackage{babel}
\usepackage{sgame}
\usepackage{epigraph}
\usepackage[normalem]{ulem}
\usepackage{cancel}
\usepackage{authblk}

 \newcommand{\noop}[1]{}

\makeatletter
\renewcommand*{\@fnsymbol}[1]{\ensuremath{\ifcase#1\or *\or @\or \ddagger\or
    \mathsection\or \mathparagraph\or \|\or **\or \dagger\dagger
    \or \ddagger\ddagger \else\@ctrerr\fi}}
\makeatother

  \providecommand{\axiomname}{Axiom}
  \providecommand{\definitionname}{Definition}
  \providecommand{\examplename}{Example}
  \providecommand{\lemmaname}{Lemma}
  \providecommand{\remarkname}{Remark}
\providecommand{\corollaryname}{Corollary}
\providecommand{\theoremname}{Theorem}
\providecommand{\notationname}{Notation}
\providecommand{\resultname}{Result}
\providecommand{\propositionname}{Proposition}
\providecommand{\assumptionname}{Assumption}
 \providecommand{\claimname}{Claim}
 \providecommand{\vect}[1]{\boldsymbol{#1}}
\RequirePackage{lineno}

\title{\large{\textbf{A Free and Fair Economy: A Game of Justice and Inclusion}}\thanks{The authors thank Frank Riedel, Sarah Auster, Jan-Henrik Steg, Niels Boissonet, Andr\'{e} Casajus, Hulya Eraslan, and seminar participants at The University of Texas Rio Grande Valley, Bielefeld University, and the Texas Economic Theory Camp for their valuable and insightful comments and suggestions. Demeze-Jouatsa gratefully acknowledges financial support from the DFG (Deutsche Forschungsgemeinschaft / German Research Foundation) via grant Ri 1128-9-1 (Open Research Area in the Social Sciences, Ambiguity in Dynamic Environments), Bielefeld Young Researchers' Fund, the BGTS Mobility Grants, and the University of Ottawa.}\vspace{-1em}}

\author{Ghislain H. Demeze-Jouatsa\thanks{Center for Mathematical Economics, University of Bielefeld; Email: demeze\_jouatsa@uni-bielefeld.de.} \ \ \ Roland Pongou\thanks{Department of Economics, University of Ottawa; Email: rpongou@uottawa.ca.} \ \ Jean-Baptiste Tondji\thanks{Department of Economics and Finance, The University of Texas Rio Grande Valley; Email: jeanbaptiste.tondji@utrgv.edu.} \vspace{-2em}}

\date{\today\\ \vspace{-2em}}

\begin{document}
\maketitle

\begin{abstract}
\small{Frequent violations of \textit{fair principles} in real-life settings raise the fundamental question of whether such principles can guarantee the existence of a self-enforcing equilibrium in a free economy. We show that elementary principles of distributive justice guarantee that a pure-strategy Nash equilibrium exists in a \textit{finite} economy where agents freely (and non-cooperatively) choose their inputs and derive utility from their pay. Chief among these principles is that: 1) your pay should not depend on your \textit{name}; and 2) a more productive agent should not earn less. When these principles are violated, an equilibrium may not exist. Moreover, we uncover an intuitive condition---\textit{technological monotonicity}---that guarantees equilibrium uniqueness and efficiency. We generalize our findings to economies with social justice and inclusion, implemented in the form of progressive taxation and redistribution, and guaranteeing a basic income to unproductive agents. Our analysis uncovers a new class of strategic form games by incorporating normative principles into non-cooperative game theory. Our results rely on no particular assumptions, and our setup is entirely non-parametric. Illustrations of the theory include applications to exchange economies, surplus distribution in a firm, contagion and self-enforcing lockdown in a networked economy, and bias in the academic peer-review system.}\newline
\\
\small{\textbf{Keywords}: Market justice; Social justice; Inclusion; Ethics; Discrimination; Self-enforcing contracts; Fairness in non-cooperative games; Pure strategy Nash equilibrium; Efficiency.
}
\end{abstract}

\epigraph{\footnotesize{\textit{``For Aristotle, justice means giving people what they deserve, giving each person his or her due."}}}
{\textit{\citet[P. 187]{sandel2010justice}}}

\section{Introduction} \label{sec:introduction}
It is generally acknowledged that \textit{justice} is the foundation of a stable, cohesive, and productive society.\footnote{The Merriam-Webster dictionary defines \textit{justice} as ``the maintenance or administration of what is just especially by the impartial adjustment of conflicting claims or the assignment of merited rewards or punishments."} However, 
violations of \textit{fair principles} are highly prevalent in real-life settings. For example, discriminations based on race, gender, culture and several other factors have been widely documented (see, for instance, \citet{reimers1983labor}, \citet{wright1991gender},  \citet{sen1992missing}, \citet{bertrand2004emily}, \citet{anderson2010missing}, \citet{pongouserrano2013}, \citet{goldin2017expanding}, \citet{bapuji2020organizations}, \citet{hyland2020gendered}, \citet{card2020referees},  and \citet{koffi2021racial}). These realities raise the fundamental question of how basic principles of justice affect individual incentives, and whether such principles can guarantee the stability and efficiency of contracts among private agents in a free and competitive economy. That the literature has remained silent on this question is a bit surprising, given the long tradition of ethical and normative principles in economic theory and the relevance of these principles to the real world \citep{sen2009idea, thomson2016fair}. 
The main goal of this paper is to address this problem.   In our treatment of this question, we incorporate elementary principles of justice and ethics into non-cooperative game theory. In doing so, we uncover a new class of strategic form games with a wide range of applications to classical and more recent economic problems.   

We precisely address the following questions:
\begin{enumerate}
\item[A:] How do \textit{fair principles} affect the stability of social interactions in a free economy?
\item[B:] Under which conditions do \textit{fair principles} lead to \textit{equilibrium efficiency}?
\end{enumerate}
To formalize these questions, we introduce a model of a \textit{free and fair economy}, where agents freely (and non-cooperatively) choose their inputs, and the surplus resulting from these input choices is shared following four elementary principles of distributive justice, which are:\\
\begin{enumerate}
\item \textbf{Anonymity}: Your pay should not depend on your \textit{name}.\footnote{Here, \textit{name} designates any unproductive individual characteristic such as first and last names, skin color, gender, religious or political affiliation, cultural background, etcetera. Anonymity means that a person's pay should not depend on their identity; in other words, given my input choice and that of others, my pay should not vary depending on whether I am called ``Emily/Greg" or ``Lakisha/Jamal" \citep{bertrand2004emily}, or depending on whether my skin color is black, white or green, or depending on whether I am a man or a woman.}
\item \textbf{Local efficiency}: No portion of the surplus generated at any profile of input choices should be wasted.
\item \textbf{Unproductivity}: An unproductive agent earns nothing.
\item \textbf{Marginality}: A more productive agent should not earn less.
\end{enumerate}
It is generally agreed that these ideals form the core principles of \textit{market} (or \textit{meritocratic}) \textit{justice}, and are of long tradition in economic theory. They have inspired eighteenth centuries writers like \citet{rousseau1895social} and \citet{aristotle1946}, and contemporary authors like  \citet{rawls1971theory}, \citet{shapley1953value}, \citet{young1985monotonic},  \citet{roemer1998theories}, \cite{de2008marginal}, \citet{sen2009idea}, \citet{sandel2010justice}, \citet{thomson2016fair}, and \citet{posner2018radical}, among several others. However, a number of empirical observations have suggested that the real world does not always conform to these elementary principles of justice. Studies have shown that \textit{anonymity} is violated in job hiring \citep{kraus2019evidence, bertrand2004emily}, in wages \citep{charles2008prejudice, lang2011education}, in scholarly publishing \citep{laband1994favoritism, ellison2002evolving, heckman2017publishing, serrano2018top5itis, akerlof2020sins, card2020referees}, in school admission \citep{francis2012redistributive,  grbic2015role}, in sexual norm enforcement \citep{pongouserrano2013}, in health care \citep{balsa2001statistical, thornicroft2007discrimination}, in household resource allocations \citep{sen1992missing, anderson2010missing}, in scholarly citations \citep{card2020referees, koffi2021gendered}, and in organizations \citep{small2020sociological, koffi2021racial}. These studies generally show that discrimination based on name, race, gender, culture, religion, and academic affiliation is prevalent in these different contexts. Violations of basic principles of justice therefore raise the fundamental question of how these principles affect individual incentives, the stability of social interactions, and economic efficiency. 

We examine these questions through the lens of a model of a free and fair economy. This model is a list $\mathcal{E}= (N, \times_{j\in N} X_j, o, f, \phi, (u_j)_{j\in N})$, where $N$ is a finite set of agents, $X_j$ a finite set of actions (or inputs) available to agent $j$, $o= (o_j)_{j\in N}$ a reference profile of actions, $f$ a production (or surplus) function (also called technology) that maps each action profile $x\in \times_{j\in N}X_j$ to a measurable output $f(x)\in \mathbb{R}$, $\phi$ an allocation scheme that distributes any realized surplus $f(x)$ to agents, and $u_j$ the utility function of agent $j$. The reference point $o$  can be interpreted as an unproduced endowment of goods (or resources) that can be either consumed as such, or may be used in the production process when production opportunities are specified. Agent $j$'s action set $X_j$ can be interpreted broadly, as we do not impose any particular structure on it other than it being finite. It may be viewed as a capability set \citep{sen2009idea}, or may represent the set of different occupations (or functions) available to agent $j$ based on agent $j$'s skills, or the set of effort levels that agent $j$ may supply in a production environment. The nature of the set of actions can also be different for each agent. For each input profile $x$, the allocation scheme $\phi$ distributes the generated surplus $f(x)$ following the aforementioned principles of \textit{anonymity}, \textit{local efficiency}, \textit{unproductivity}, and \textit{marginality}, and each agent $j$ derives utility from her payoff $u_j(x)=\phi_j(f, x)$.\footnote{The formalization of these principles differ depending on the context. Ours is a generalization of the classical formalization of \citet{shapley1953value} and \citet{young1985monotonic} to our economic environment. Indeed, we show that these four principles uniquely characterize a pay scheme that generalizes the classical Shapley value (Proposition \ref{uniqueshapley}). This pay scheme is a multivariate function defined at each input profile $x$; see also \citet{PongouTondji2018} and \citet{aguiar2018non, aguiar2020index}. Also, note that $u_j(x)$ can be any increasing function of the payoff $\phi_j(f, x)$, and the functional form might be different for each agent.}

To define an equilibrium concept that captures individuals' incentives in a free and fair economy, we first observe that any economy $\mathcal{E}$ induces a corresponding strategic form game $G^{\mathcal{E}}= (N, \times_{j\in N} X_j, (u_j)_{j\in N})$.\footnote{The class of free and fair economies therefore defines a large class of games that can be characterized as \textit{fair}. Any strategic form game is either fair or unfair, and some unfair games are simply a monotonic transformation of fair games.}  Then, a profile of actions $x^{*}\in \times_{j\in N} X_j$ is said to be an \textit{equilibrium} in the free and fair economy $\mathcal{E}$ if and only if it is a \textit{pure strategy} Nash equilibrium of the game $G^{\mathcal{E}}$. 


Our first main result shows that the principles of market justice stated above guarantee the existence of an equilibrium (Theorem \ref{theorem: non emptyness: shapley}). Moreover, when an economy violates these principles, an equilibrium may not exist. These findings have profound implications. One implication is that fair rules guarantee the existence of self-enforcing contracts between private agents in a free economy. A second implication is that fair rules prevent output (and income) volatility, given that action choices at equilibrium are \text{pure} strategies. Moreover, from a purely theoretical viewpoint, the incorporation of normative principles into non-cooperative game theory has led us to identify an interesting class of strategic form games that always have a pure strategy Nash equilibrium in spite of the fact that each player has a \textit{finite} action set.\footnote{As is well known, a pure strategy Nash equilibrium does not exist in a finite strategic form game in general \citep{nash1951non}. A growing literature seeks to identify conditions under which a pure strategy Nash equilibrium exists in a finite game (see, for example, \citet{rosenthal1973class}, \citet{monderer1996potential}, \citet{mallick2011existence}, \citet{carmona2020pure}, and the references therein). But unlike our paper, this literature has not approached this problem from a normative perspective. We therefore view our analysis as a contribution.}    

Although a pure strategy equilibrium always exists in any free and fair economy, this equilibrium may be inefficient. We uncover a simple structural condition that guarantees equilibrium efficiency. More precisely, we show that if the technology is \textit{strictly monotonic}, there exists a unique equilibrium, and this equilibrium is Pareto-efficient (Theorem \ref{monotonicresult2}). Quite interestingly, we find that when a monotonic economy fails to satisfy the principles of market justice, even if an equilibrium exists, it may be inefficient.\footnote{A clear example is the prisoner's dilemma game. Economies that are modeled by such games are monotonic, although their unique equilibrium is Pareto-inefficient.} 
A clear implication of this finding is that in the class of monotonic economies, any allocation scheme that violates the principles of market justice is welfare-inferior to the unique scheme that respects these principles.

Next, we extend our analysis to economies with \textit{social justice}. The principles of \text{market justice} imply that unproductive agents (for example, agents with severe disabilities) should earn nothing. In most societies, however, social security benefits ensure that a basic income is allocated to agents who, for certain reasons, cannot produce as much as they would like to (see, for example, among others, \citet{david2006growth}, and \citet{hanna2018universal}). To account for this reality, we extend our model to incorporate \textit{social justice} or \textit{inclusion}. Generally, social justice includes solidarity and moral principles that individuals have equal access to social rights and opportunities, and it requires consideration beyond talents and skills since some agents have natural limitations, not allowing them to be productive. 

Social justice is incorporated into our model in the form of progressive taxation and redistribution. At any production choice, a positive fraction of output is taxed and shared equally among all agents, and the remaining fraction is allocated according to the principles of market justice. This allocation scheme satisfies the principles of \textit{anonymity} and \textit{local efficiency}, but violates \textit{marginality} and \textit{unproductivity}. Income is redistributed from the high skilled and talented (or more productive agents) to the least well-off. However, the income rank of a free and fair economy (without social justice) is maintained, provided that the entire surplus is not taxed. We generalize each of our results. In particular, a pure strategy equilibrium always exists regardless of the tax rate (Theorem \ref{theo:Egalitarian Shapley value}). Consistent with Theorem \ref{monotonicresult2}, we also find that if the production technology is \textit{strictly monotonic}, there exists a unique equilibrium, and this equilibrium is Pareto-efficient (Corollary \ref{monotonicresultsocialjustice}).

We uncover additional results on the efficiency of economies with social justice. In particular, we find that there exists a tax rate threshold above which there exists a pure strategy Nash equilibrium that is Pareto-efficient, even if the economy is not monotonic (Theorem \ref{theo:Egalitarian Shapley value efficiency}). Moreover, we show that one can always change the reference point of any \textit{non-monotonic} free economy with social justice to guarantee the existence of an equilibrium that is Pareto-efficient (Theorem \ref{res:referencepoint}). This latter finding implies that if a free economy is able to choose its reference point, then it can always do so to induce a Pareto-efficient outcome that is self-enforcing.


We develop various applications of our model to classical and more recent economic problems. In particular, we develop applications to exchange economies \citep{walras1954elements, arrow1954existence, shapley1977example, osborne1994course}, surplus distribution in a firm, self-enforcing lockdown in a networked economy with contagion, and bias in the academic peer-review system \citep{akerlof2020sins}. This variety of applications is possible because we impose no particular assumptions on the structure of action sets, and the action set of each agent may be of a different nature. We start with applying our theory to a production environment where an owner of the firm (or team leader) uses bonuses as a device to incentivize costly labor supply from rational workers. Our analysis shows that in addition to guaranteeing equilibrium existence, the owner can also achieve production efficiency, provided that the costs of labor supply are not too high. Next, we provide an application to contagion in a networked economy in which rational agents freely form and sever bilateral relationships. Rationality is captured by the concept of pairwise-Nash equilibrium, which refines the Nash equilibrium. Using a contagion index \citep{pongouserrano2013}, we show how the costs of a pandemic can induce self-enforcing lockdown. Our application to academic peer-review in the knowledge economy shows that discrimination in the allocation of rewards results in a Pareto-inferior outcome, which indicates that bias reduces the incentive to study ``soft", ``important", and relevant topics in equilibrium.\footnote{See, for example, a recent study by \citet{akerlof2020sins} on the consequences of mostly rewarding ``hard" research topics in the field of economics.} Finally, we recast the model of an exchange economy in our framework, and show that our equilibrium is generally different from the Walrasian equilibrium. This difference is in part explained by the fact that the Walrasian model assumes linear pricing, whereas our model is fully non-parametric.  


The rest of this paper is organized as follows.  Section \ref{sec:preliminary} introduces the model of a free and fair economy. In Section \ref{sec:faireconomyandequilibrium}, we prove the existence of a pure strategy Nash equilibrium in a free and fair economy. Section \ref{sec:efficiency} is devoted to the analysis of efficiency. In Section \ref{sec:socialjustice}, we extend our model to incorporate social justice and inclusion, and we generalize our results. In Section \ref{sec:application}, we present some applications of our analysis. Section \ref{sec:literature} situates our paper in the closely related literature, and  Section \ref{conclusion} concludes. Some proofs are collected in an appendix.  




\section{A free and fair economy: definition, existence and uniqueness}\label{sec:preliminary}
 
\ \ \ In this section, we introduce preliminary definitions and the key concepts of the paper. We then show that there exists a unique economy that is free and fair.
 
 \subsection{A free economy} \label{sec:freeconomy}
 
 \ \ \ A free economy is an economy where agents freely choose their actions and derive utility from their pay. It is modeled as a list $\mathcal{E}= (N, \times_{j\in N} X_j, (o_j)_{j\in N}, f, \phi, (u_j)_{j\in N})$. $N=\{1,2,..., n\}$ is a finite set of agents. Each agent $j$ has a finite set of feasible actions $X_j$. We refer to an action profile $x=(x_j)_{j\in N}$ as an \textit{outcome}, and denote the set $\times_{j\in N} X_j$ of outcomes by $X$. The \textit{reference outcome} (also called \textit{reference point}) is $o=(o_j)_{j\in N}$; it can be interpreted as the \textit{inaction} point, where agents do nothing or do not engage in any sort of transactions with other agents. A production (or surplus) function (also called technology) $f$ transforms any choice $x$ to a real number $f(x) \in \mathbb{R}$, with $f(o)=0$.\footnote{We normalize the surplus at the reference point to $0$ for expositional purposes. It is possible that the surplus realized at $o$ is not zero, and in this case, $f(x)$ should be interpreted as \textit{net surplus} at $x$, that is, the realized surplus at $x$ minus the realized surplus at $o$. We assume the reference $o$ to be exogenously determined.}  We denote by  $P(X) = \{g:X\rightarrow \mathbb{R}, \ \text{with} \ g(o)=0\}$ the set of production functions  on $X$. $\phi: P(X) \times X \rightarrow \mathbb{R}^n$ is a distribution scheme that assigns to each pair $(f, x)$ a payoff vector $\phi(f,x)$. At each input profile $x$, each agent $j$ derives utility $u_j(x)=\phi_j(f,x)$.\footnote{As noted in the Introduction, $u_j(x)$ can be any increasing function of $\phi_j(f,x)$, where the functional form may be different for each agent.}


\vspace{2mm}


\subsection{A free and fair economy} \label{sec:faireconomy}
 \ \ \ A free and fair economy is a free economy $\mathcal{E}= (N, \times_{j\in N} X_j, (o_j)_{j\in N}, f, \phi, (u_j)_{j\in N})$ in which the surplus distribution scheme $\phi$ satisfies elementary principles of market justice. These principles, of long tradition in economic theory, are those of \textit{anonymity}, \textit{local efficiency}, \textit{unproductivity}, and \textit{marginality} stated in the Introduction. These principles are naturally interpreted, but their formalization varies depending on the context. A few preliminary definitions and notations will be needed for their formalization in our setting.
 
 \begin{definition}
 Let $x\in X$ a profile of actions. An outcome $x'\in X$ is a \textit{sub-profile} of $x$ if either $x'=x$ or $[x'_{i} \neq x_i \Longrightarrow x'_i = o_i]$, for $i\in N$. 
 \end{definition}
 
For each $x\in X$, we denote by $\Delta (x)$ the set of sub-profiles of $x$. Given a production function $f\in P(X)$, and an outcome $x\in X$, we define the function $f^x$ as the restriction of $f$ to $\Delta(x)$:
\begin{equation*}
f^x: \Delta(x) \rightarrow \mathbb{R}, \ \text{such that} \ f^x(y)=f(y), \ \text{for each} \ y\in \Delta(x).   
\end{equation*}

\begin{definition}
Let $i \in N$. We define the relation $\Delta_o^{i}$ on $X$ by: \[ [x' \ \Delta_o^{i} \ x] \ \text{if and only if}   \ [x'\in \Delta(x) \ \text{and} \ x'_i=o_i]. \]
\end{definition}



Let  $x\in X$ be an outcome. We denote $\Delta_o^{i}(x)=\{x'\in X: x' \ \Delta_o^{i} \ x\}$, and by $N^x=\{i\in N: x_i \neq o_i\}$ the set of agents whose actions in $x$ are different from their reference points. We also denote $|x| = |N^x|$ the cardinality of $N^x$.

\begin{definition}
Let $f\in P(X)$, $x \in X$, and $x' \in \Delta_o^{i}(x)$. The \textit{marginal contribution} of agent $i$ at a pair $(x', x)$ is: 
\[mc_i(f, x', x)= f(x'_{-i},x_i)-f(x'),\] 
where $(x'_{-i}, x_i)\in X$ is the outcome in which agent $i$ chooses $x_{i}$, and every other agent $j$ chooses $x'_{j}$. 
\end{definition}


\begin{definition}
Let $f \in P(X)$. Agent $i$ is said to be \textit{unproductive} if for each $x\in X$ and all $x'\in \Delta^{i}_{0} (x)$, $mc_i(f, x', x)=0$.
\end{definition}


A permutation $\pi$ of $N$ is a bijection of $N$ into itself. We denote by $\mathcal{S}_n$ the set of permutations of $N$. Let $x\in X$ be a profile of inputs, and let $\pi^x \in \mathcal{S}_n$ be a permutation of $N$ whose restriction to $N\backslash N^x$ is the identity function, that is $\pi^x(i)=i$ for each $i\in N\backslash N^x$. Remark that $\pi^x$ permutes only agents that are active in the profile $x$, and is therefore equivalent to a permutation $\pi^x:N^x\rightarrow N^x$ over $N^x$; we denote by $\mathcal{S}_n^{x}$ the set of such permutations.   

Let $x \in X$, $\pi^x \in \mathcal{S}_n^{x}$, and $y\in \Delta(x)$. We define the profile $\pi^x(y) = (\pi^x_j(y))_{j\in N}$, where
$$\pi^x_j(y)=\left\{ \begin{array}{ll}
    x_j & \text{ if } y_{k}\neq o_{k}, \ j = \pi^x(k) \\
    o_j & \text{ if } y_{k} = o_{k}, \ j = \pi^x(k).
\end{array}\right. $$


We now formalize the principles of market justice below. 
\begin{description}



\item \textbf{Anonymity}. An allocation $\phi$ satisfies $x-$Anonymity if for each $i\in N$ and $\pi^x \in \mathcal{S}_n^{x}$,
\begin{equation*}
\phi_i(\pi^x f^x, x) = \phi_{\pi^x(i)} (f^x, x), \ \text{where} \ \pi^x f^x (y) = f^x(\pi^x(y)), \ \text{for} \ y \in \Delta (x).    
\end{equation*}
The value $\phi$ satisfies Anonymity if $\phi$ satisfies $x-$Anonymity for all $x\in X$. 

\item \textbf{Local Efficiency}. $\sum\limits _{j\in N}\phi_{j}(f,x)= f(x)$ for any $f \in P(X)$ and $x \in X$. 

\item \textbf{Unproductivity.} If agent  $i$ is unproductive, then $\phi_i(f, x) = 0$ for each $f\in P(X)$ and $x \in X$.

\item \textbf{Marginality}. Let $f, g \in P(X)$, and $x$ an outcome. If
\begin{equation*}
mc_i(f,x',x) \geq mc_i(g,x',x) \  \text{for each} \ x' \in \Delta_o^{i}(x)
\end{equation*}
for an agent $i$, then $\phi_{i}(f,x)\geq \phi_{i}(g,x)$. 
\end{description}

These axioms are interpreted naturally. \textit{Anonymity} means that an agent’s pay does not depend on their \textit{name}. It states that every agent is treated the same way by the allocation rule: if two agents exchange their identities, their payoffs will remain unchanged. An important property that is implied by anonymity is \textit{symmetry} (or \textit{non-favoritism}), which means that equally productive agents should receive the same pay. \textit{Local efficiency} simply requires that the surplus resulting from any input choice be fully shared among productive agents participating in the economy. \textit{Unproductivity} means that an agent whose marginal contribution is zero at an input profile should get nothing at that profile. \textit{Marginality} means that, if the adoption of a new technology increases the marginal contribution of an agent, that agent's pay should not be lower under this new technology relative to the old technology. In other words, more productive agents should not earn less compared to less productive agents. Throughout the paper, we abbreviate the four principles as \textbf{ALUM}.

\begin{definition}
A \textbf{free and fair economy} is a free economy $(N, X,o, f, \phi, u)$ such that the distribution scheme $\phi$ satisfies \textbf{ALUM}. 
\end{definition} 

We have the following result.

\begin{proposition}\label{uniqueshapley}
There exists a unique distribution scheme, denoted $\vect{Sh}$, that satisfies \textbf{ALUM}. For  any production function $f \in P(X)$, and any given outcome $x \in X$ and agent $i \in N$:
\begin{equation}
\vect{Sh}_{i}(f,x)=\sum\limits_{x' \in  \Delta_{o}^{i}(x)}\frac{(|x'|)!(|x|-|x'|-1)!}{(|x|)!} \ mc_i(f, x', x).  \label{Shapleyfunction}
\end{equation}
\end{proposition}

\begin{proof}[Proof of Proposition \ref{uniqueshapley}]
See Appendix.
\end{proof}

Remark that for each agent $i$, the value $\vect{Sh}_i(f,x)$ is interpreted as agent $i$'s average contribution to output $f(x)$. It can be easily shown that the allocation rule $\vect{Sh}$ generalizes the classical Shapley value \citep{shapley1953value}. In fact, to obtain the classical Shapley value, one only has to assume that each agent's action set is the pair $\{0, 1\}$; the classical Shapley value is simply $\vect{Sh}_i(f,x)$ where $x=(1,1,...,1)$, which effectively corresponds to the assumption that the grand coalition is formed. Our setting generalizes the classical environment in three ways. First, it is not necessary to assume that all players have the same action set. Second, the action set of a player may have more than two elements. Third, the value can be computed for any input profile $x$, which effectively means that $\vect{Sh}_i(f,x)$ as a multivariate function of $x$. Our model also generalizes that in \citet{PongouTondji2018} (when the environment is certain), \citet{aguiar2018non}, and \citet{aguiar2020index}. Following these latter studies, we will call $\vect{Sh}$ the \textit{Shapley pay scheme}.  

Below, we illustrate the notion of a free and fair economy, and provide an example of a free economy that is unfair.

\begin{example} \label{introexample}
Consider a small economy $\mathcal{E}=(N, X, o, f, \phi, u)$, where $N=\{1, 2\}$, $X_1= \{a_1, a_2\}$, $X_2=\{b_1,b_2, b_3\}$, $o=(a_1, b_1)$, $X= X_1 \times X_2$, $f$ is given by $f(a_1,b_1)= 0, \ f(a_1,b_2)=5=f(a_1,b_3), \  f(a_2, b_1)=2, \ \text{and} \ f(a_2, b_2)=4=f(a_2,b_3)$, and for each $x \in X$, $\phi (f, x)=u (f, x)$ is given in Table \ref{Tablenfair2} below:

\begin{table}[htb]\hspace*{\fill}%
\begin{game}{2}{3}[Agent 1][Agent 2]
& $b_1$ & $b_2$ & $b_3$\\
$a_1$ &$(0,0)$ &$(0,5)$ &$(0,5)$\\
$a_2$ &$(2,0)$ &$(0.5,3.5)$ &$(0.5,3.5)$
\end{game}\hspace*{\fill}
\caption{A 2-agent free and fair economy}\label{Tablenfair2}
\end{table}

For each of the six payoff vectors presented in Table \ref{Tablenfair2}, the first component represents agent 1's payoff (for example, $u_{1}(f, (a_2, b_1))= 2$) and the second component represents  agent 2's payoff (for instance, $u_{2}(f, (a_2, b_1))= 0$). We can check that for each $x \in X$, $u (f, x) = \phi (f, x)= \vect{Sh}(f,x)$. Therefore, $\mathcal{E}$ is a free and fair economy.

\begin{table}[htb]\hspace*{\fill}%
\begin{game}{2}{3}[Agent 1][Agent 2]
& $b_1$ & $b_2$ & $b_3$\\
$a_1$ &$(0,0)$ &$(2,3)$ &$(3,2)$\\
$a_2$ &$(1,1)$ &$(3,1)$ &$(2,2)$
\end{game}\hspace*{\fill}
\caption{A 2-agent free and unfair economy}\label{Tablenfair1}
\end{table}

Now, we consider another economy $\mathcal{E}'$ with the same characteristics as in $\mathcal{E}$ except for the distribution scheme $\phi$ that is replaced by a new scheme $\psi$ described in Table \ref{Tablenfair1}. In addition to the fact that $\psi \neq \vect{Sh}$, it is straightforward to show that the  distribution $\psi$ violates the \textit{marginality} axiom. Therefore,  $\mathcal{E}'$ is not a free and fair economy. 

One of our goals in this paper is to answer the question of whether \textit{fair principles} guarantee the existence of a pure strategy Nash equilibrium. We can observe that in the free and fair economy described by Table \ref{Tablenfair2}, there are two pure strategy Nash equilibria, which are $(a_2, b_2)$ and $(a_2, b_3)$. However, the modified economy $\mathcal{E}'$ represented by Table \ref{Tablenfair1} admits no equilibrium in pure strategies. In the next section, we will show that fair principles guarantee the existence of a pure strategy Nash equilibrium in a free economy, and when an economy violates these principles, a pure strategy Nash equilibrium may not exist. 
\end{example}

\section{Equilibrium existence in a free and fair economy}\label{sec:faireconomyandequilibrium}

\ \ \ In a free and fair economy, agents make decisions that affect their payoff and the payoffs of other agents. One natural question that therefore arises is whether an \textit{equilibrium} exists. In this section, we first show that a free economy can be modeled as a strategic form game and use the notion of pure strategy Nash equilibrium \citep{nash1951non} to capture incentives and rationality. Our main result is that a free and fair economy always has a pure strategy Nash equilibrium. 

\subsection{A free and fair economy as a strategic form game}\label{fair-strategygames}

\ \ \ A \textit{strategic form game} is a 3-tuple $(N,X,v)$, where $N$ is the set of players, $X= \times_{j \in N} X_j$ is the strategy space, and $v: X \rightarrow \mathbb{R}^n$ is the payoff function. For each $x\in X$, $v_i(x)$ is agent $i$'s payoff at strategy profile $x$, for each $i \in N$. A strategic form game is said to be finite if the set of agents $N$ is finite, and for each agent $i$, the set  $X_i$ of actions is also finite.

A strategy profile $x^{*} \in X$ is a \textit{pure strategy Nash equilibrium} in the game $(N,X,v)$ if and only if for all $i\in N$, $v_i(x^{*}) \geq v_i(x^{*}_{-i},y_i)$, for all $y_i\in X_i$, where $(x^{*}_{-i},y_i)$ is the strategy profile in which agent $i$ chooses $y_{i}$ and every other agent $j$ chooses $x^{*}_{j}$.

A free economy $\mathcal{E}=(N,X, o, f, \phi, u)$ generates a strategic form game $G^{\mathcal{E}}=(N, X, u^{\mathcal{E}})$, where for each $x\in X$ and each $i\in N$, $u_i^{\mathcal{E}}(x)= u_i(f,x)= \phi_i(f,x)$. In the case $\mathcal{E}$ is a free and fair economy, then for each outcome $x$, $\sum \limits_{j\in N} u_j^{\mathcal{E}}(x)= f(x)$ since the distribution scheme $\phi$ satisfies \textit{local efficiency}. For this reason, when $\mathcal{E}$ is a free and fair economy, we may refer to the production function $f$ as the \textit{total utility function} of the strategic form game $G^{\mathcal{E}}$. 

 \begin{definition}
Let $\mathcal{E}=(N, X, o, f, \phi, u)$ be a free economy. A profile  $x^{*} \in X$ is an \textit{equilibrium}  if and only if $x^{*}$ is a pure-strategy Nash equilibrium in the strategic form game $G^{\mathcal{E}}$.
\end{definition}

\subsection{Existence of an equilibrium}\label{sec:existence}

\ \ \ In this section, we state and prove our main result.

\begin{theorem}
\label{theorem: non emptyness: shapley} Any free and fair economy $\mathcal{E}=(N, X, o, f, \phi, u)$  admits an equilibrium.
\end{theorem}

The proof of Theorem \ref{theorem: non emptyness: shapley} uses the concept of a \textit{cycle of deviations} that we introduce below.

\begin{definition}\label{cyle}
Let $G=(N,X,v)$ be a strategic form game and $L^{k}=(x^{1}, x^2, ..., x^{k})$ be a list of outcomes, where each  $x^l \in X$  ($l=1,..., k$) is a pure strategy. The $k$-tuple $L^{k}$ is a \textit{cycle of deviations} if there exist agents $j_{1},...,j_{k}\in N$ such that
\begin{equation*}
x^{l+1}=(x^{l}_{-j_l},x^{l+1}_{j_l}) \ \text{and} \ v_{j_{l}}(x^{l+1})>v_{j_{l}}(x^{l})
\end{equation*}
for each $l=1,...,k$, and $x^{k+1}=x^1$.
\end{definition}

\begin{example} \label{game: game with no cycle of deviation and shapley}
In the strategic form game represented in Table \ref{Table5}, consider the list $L^4 = (x^1, x^2, x^3, x^4)$, where  $x^1=(c,a)$, $x^2=(d,a)$, $x^3=(d,b)$, and $x^4=(c,b)$. 
\begin{table}[htb]\hspace*{\fill}%
\begin{game}{2}{2}[Agent~1][Agent~2]
& $a$ & $b$ \\
$c$ &$(0,4)$ &$(3,0)$ \\
$d$ &$(1,0)$ &$(0,2)$
\end{game}\hspace*{\fill}%
\caption{A 2-agent game that admits a cycle of deviations}\label{Table5}
\end{table}

$L^4$ forms a cycle of deviations. Indeed, agent 1 has an incentive to deviate from $x^1$ to $x^2$. By doing so, agent 1 receives an excess payoff of $1$. Similarly, agent 2 receives an excess payoff of 2 by deviating from $x^2$ to $x^3$. Agent 1 receives an excess payoff of 3 by deviating from $x^3$ to $x^4$; and agent 2 receives an excess payoff of 4 by deviating from $x^4$ to $x^1$. The sum of excess payoffs in the cycle $L^4$ is therefore equal to $10$. 

\begin{table}[htb]\hspace*{\fill}%
\begin{game}{3}{4}[Agent~1][Agent~2]
& $b_1$ & $b_2$ & $b_3$ & $b_4$\\
$a_1$ &$(0,0)$ &$(0,0)$ &$(0,12)$ & $(0,6)$\\
$a_2$ &$(13,0)$ &$(\frac{13}{2}, -\frac{13}{2})$ & $(\frac{3}{2}, \frac{1}{2})$ & $(4,-3)$\\
$a_3$ &$(3,0)$ &$(8,5)$ &$(-1,8)$ & $(-1,2)$
\end{game}\hspace*{\fill}%
\caption{A 2-agent game with Shapley payoffs}
\label{table: game with Shapley payoffs}
\end{table}

In the strategic form game in Table \ref{table: game with Shapley payoffs}, the sum of excess payoffs in any cycle of outcomes equals 0. Therefore, the game does not admit a cycle of deviations. The profile $x^{*}= (a_2, b_3)$ is the only pure strategy Nash equilibrium of the game.
\end{example}

Note that the game in Table \ref{table: game with Shapley payoffs} is generated from a free and fair economy. From Definition \ref{cyle}, a sufficient condition for a finite strategic form game to admit a pure strategy Nash equilibrium is the absence of a cycle of deviations. The sum of excess payoffs in any cycle of deviations has to be strictly positive, as illustrated in Table \ref{Table5} in Example \ref{game: game with no cycle of deviation and shapley}. Such an example of a cycle of deviations can not be constructed in a strategic form game generated from a free and fair economy (see Table \ref{table: game with Shapley payoffs} in Example \ref{game: game with no cycle of deviation and shapley}). We  prove that in a strategic form game generated by a free and fair economy, the sum of excess payoffs in any cycle of deviations equals $0$.

\begin{lemma}\label{lemma:deviation cycle}

Let $\mathcal{E}=(N,X, o, f, \phi, u)$ be a free and fair economy, and $G^{\mathcal{E}}=(N, X, u^{\mathcal{E}})$ the strategic form  game generated by $\mathcal{E}$. Then, the sum of \textit{excess payoffs} in any cycle of deviations in $G^{\mathcal{E}}$ equals $0$. 
\end{lemma}

\begin{proof}[Proof of Lemma \ref{lemma:deviation cycle}.]
In this proof, we simply denote the payoff function $u^{\mathcal{E}}$ by $u$. Let $L^k= (x^{1}, x^2, ..., x^{k})$ be a cycle of deviations in the game $G^{\mathcal{E}}$, and let agents $j_{1},...,j_{k}\in N$ be the associated sequences of defeaters. We denote by $S(L^k, u)$ the sum of excess payoffs in the cycle $L^k$: 
\begin{equation*}
S(L^k, u) = u_{j_{k}}(x^{1})-u_{j_{k}}(x^{k})+
\sum_{l=1}^{k-1}[u_{j_{l}}(x^{l+1})-u_{j_{l}}(x^{l})].
\end{equation*}
We show that in the game $G^{\mathcal{E}}$
\begin{equation*}
S(L^k,u)=0.
\label{equation, toProve}
\end{equation*}

For each agent $i \in N$, let $\mathcal{R}_i$ be a total order on the set $X_i$ such that $o_i\mathcal{R}_i x_i$ for all $x_i \in X_i$. For each outcome $x\in X$, define 
\begin{equation*}
f_x(T,y)=\left\{
\begin{array}{l}
\left\vert N^x\right\vert \text{ if }N^x\subseteq T  \text{ and } x_i\mathcal{R}_i y_i  \text{ for all  } i \in N^x\\
0\text{ otherwise}%
\end{array}%
\right. 
\end{equation*}
for all $T\subseteq N$ and $y\in X$. 

We also define  the following production function:
\begin{equation*}
    f_x(z)=f_x(N^z,z) \ \text{for all} \ z\in X.
\end{equation*}
Note that the family $\{f_x, x\in X\backslash \{o\} \}$ forms a basis of the set of production functions on the the same set of players $N$, same set of outcomes $X$, and same reference outcome $o$. Therefore, there exists $(\alpha_x)_{x\in X\backslash\{o\}}$ such that 
\begin{equation}\label{basis}
    f(z)=\sum_{x \in X} \alpha_x f_x(z) \ \text{for all} \ z\in X.
\end{equation}

Furthermore, each $f_x$, $x\in X$, is the total utility function of a strategic form game with Shapley utilities $G^x=(N,X,v^x)$, where for each $i\in N$, $v^x_i$ is given by
\begin{equation*}
    v_i^x(z)=\left\{
\begin{array}{l}
1\text{ if }i\in N^x, \ N^x\subseteq N^z, \ x_j \mathcal{R}_j z_{j}\text{ for all } j \in N^x\\
0\text{ otherwise.}%
\end{array}%
\right. \ \text{for all} \ z \in X.
\end{equation*}

Step 1. We show that the sum of excess payoffs of the cycle $L^k$ equals $0$ in each strategic form game $G^x$. First observe that $v_i^x\equiv 0$ for all $i\notin N^x$, and $v_i^x\equiv v_j^x$ for all $i,j\in N^x$. This means that the sum of excess payoffs in any cycle of the game $G^x$, and in particular in the cycle $L^k$, equals the sum of excess payoffs of any $i\in N^x$, which is obviously $0$.

Step 2. We show that $S(L^k, u)=0$.

Using equation (\ref{basis}),  $f=\sum\limits_{x\in X}\alpha_x f_x$, we have that $u=\sum\limits_{x\in X}\alpha_xv^x$. Given that $S(L^k, v^x)=0$ for each outcome $x$, we can deduce that $S(L^k, u)=0$.
\end{proof}

Now, we derive the proof of Theorem \ref{theorem: non emptyness: shapley}. 

\begin{proof}[Proof of Theorem \ref{theorem: non emptyness: shapley}] From Lemma \ref{lemma:deviation cycle}, the game  $G^{\mathcal{E}}$ admits no cycle of deviations. As $G^{\mathcal{E}}$ is finite, we conclude that $G^{\mathcal{E}}$ admits a pure strategy Nash equilibrium.
\end{proof}

The principles of market justice that define a free and fair economy are only  sufficient conditions for the existence of a pure strategy Nash equilibrium. However, an economy that violates the fair principles may not have a pure strategy Nash equilibrium.



\section{Equilibrium efficiency in a free and fair economy}\label{sec:efficiency}

\ \ \ In Section \ref{sec:existence}, we prove the existence of a pure strategy equilibrium (Theorem \ref{theorem: non emptyness: shapley}) in a free and fair economy. However, there is no guarantee that each  equilibrium is Pareto-efficient. For instance, consider the  strategic form game described in Table \ref{table: game with Shapley payoffs} in Example \ref{game: game with no cycle of deviation and shapley}. The game admits a unique pure strategy Nash equilibrium $x^{*}= (a_2, b_3)$ with $\vect{Sh} (f, x^{*})= (\frac{3}{2}, \frac{1}{2})$. However, the equilibrium $x^{*}$ is Pareto-dominated by the strategy $x =(a_3, b_2)$ with $\vect{Sh} (f, x)= (8, 5)$. Below, we provide two conditions on the production function that address this issue. The first condition---\textit{weak monotonicity}---guarantees the existence of a Pareto-efficient equilibrium in a free and fair economy, and the second condition---\textit{strict monotonicity}---guarantees that there is a unique equilibrium  and that this equilibrium is Pareto-efficient.  Importantly, we also find that in a free economy that is not fair, these monotonicity conditions do not guarantee the existence of an equilibrium that is Pareto-efficient. Before presenting these results, we need some  definitions. 

Let $\mathcal{E}=(N,X, o, f, \phi, u)$ be a free economy, and for $i \in N$, we denote $X_{-i}=\prod\limits_{j=1, \ j\neq i}^{n}X_{j}$. 
\begin{definition}
An order $R$ defined on $X$ is \textit{semi-complete} if for all $i \in N$ and $x_{-i}\in X_{-i}$, the restriction of $R$ to $A_i$ is complete, where $A_i=\{x_{-i}\}\times X_i$.
\end{definition}

\begin{definition}\label{def:monotonicity} $f \in P(X)$ is:
\begin{enumerate}
  \item  \textit{weakly monotonic} if there exists a semi-complete order $R$ on $X$ such that for any $x, y\in X$, if $x \ \textit{R} \ y$, then $f(x)\leq f(y)$.
  \item  \textit{strictly monotonic}  if there exists a semi-complete order $R$ on $X$ such that  for any $x, y\in X$, $[x \ \textit{R} \ y \text{ and }x \neq y]$ implies $f(x)< f(y)$. 
 \end{enumerate}
\end{definition}

\begin{definition}
A free and fair economy $\mathcal{E}=(N,X,o, f, \phi, u)$ is weakly (resp. strictly) monotonic if  $f$ is weakly (resp. strictly) monotonic.
\end{definition}

 We have the following result.
 
\begin{theorem}
\label{monotonicresult2}
A weakly monotonic free and fair economy $\mathcal{E}=(N,X,o, f, \phi, u)$  admits an equilibrium that is Pareto-efficient. If $\mathcal{E}$ is strictly monotonic, then, the equilibrium is unique and Pareto-efficient. 
\end{theorem}

\begin{proof}[Proof of Theorem  \ref{monotonicresult2}]
The result in Theorem \ref{monotonicresult2}  follows from the fact that each agent $i$'s payoff $\vect{Sh}_i(f, x)$ at $x$ depends only on the marginal contributions $\{f(y_{-i},x_i)-f(y), \ y \in \Delta_o^{i}(x)\}$ of that agent at $x$. Since $f$ is weakly monotonic, the underlying semi-complete relation, say $R$, satisfies the following condition: there exists $\overline{x}\in X$ such that $f$ reaches its maximum at $\overline{x}$, and for all $i\in N$ and  $x_{-i} \in X_{-i}$, we have $x \ R \ (x_{-i},\overline{x}_i)$.  Therefore, each marginal contribution of agent $i$ at a given outcome $x$ is less than or equal to his or her corresponding marginal contribution at the outcome $(x_{-i},\overline{x}_i)$. Given that the Shapley distribution scheme, $\vect{Sh}(f, .)$, is increasing in marginal contributions, agent $i$'s choice $\overline{x}_i$  is a weakly dominant strategy of agent $i$ in the game $G^{\mathcal{E}}$. Therefore, $\overline{x}$ is a Nash equilibrium. The profile $\overline{x}$ is also Pareto-efficient as it maximizes $f$. If $f$ is strictly monotonic, then each $\overline{x}_i$ is strictly dominant and $\overline{x}$ is the unique Nash equilibrium of the game $G^{\mathcal{E}}$.
\end{proof}

Theorem \ref{monotonicresult2} ensures the uniqueness and Pareto-efficiency of the equilibrium in a strictly monotonic free and fair economy. The strategic  form game described in Table \ref{table: game with Shapley payoffs} admits the profile $x^{*}=(a_2,b_3)$  as the only pure strategy Nash equilibrium. However, $x^{*}$ is Pareto-dominated by the profile $x=(a_3,b_2)$, which is not an equilibrium. Such a result can not arise in a strictly monotonic free and fair economy. In addition to providing a condition that guarantees the existence of a Pareto-efficient equilibrium, Theorem \ref{monotonicresult2} also provide a condition that rules out multiplicity of equilibria in the domain of free and fair economies.

In Theorem  \ref{monotonicresult2}, we show that each weakly monotonic free and fair economy admits an equilibrium that is Pareto-efficient. Consider the  strategic form game described in Table \ref{Tableeff} below. The latter is derived from a free and fair economy with the profile $o= (c,a)$ as the reference point. The economy admits two equilibria, namely, outcomes $(c,a)$ and $(d,b)$. The profile $(d,b)$ is Pareto-efficient and it dominates the outcome $(c,a)$.
\begin{table}[htb]\hspace*{\fill}%
\begin{game}{2}{2}[Agent~1][Agent~2]
& $a$ & $b$ \\
$c$ &$(0,0)$ &$(0,0)$ \\
$d$ &$(0,0)$ &$(1, 1)$
\end{game}\hspace*{\fill}%
\caption{A 2-agent free and fair economy with a Pareto-dominated equilibrium}\label{Tableeff}
\end{table}

\begin{table}[htb]\hspace*{\fill}
\begin{game}{2}{2}[Agent~1][Agent~2]
& $b_1$ & $b_2$ \\
$a_1$ &$(0,0)$ &$(2,-1)$ \\
$a_2$ &$(2,0)$ &$(1,2)$
\end{game}\hspace*{\fill}
\caption{A 2-agent strictly monotonic free and unfair economy}\label{Table3}
\end{table}

We relate the existence of an equilibrium that is Pareto-dominated in the free and fair economy described in Table \ref{Tableeff} to the fact that the production function is weakly monotonic. However, it is essential to emphasize that the existence of an equilibrium is due to the fact that the economy is fair and not to the monotonicity property of the technology. For instance, consider a free economy $\mathcal{E}^f$, where agents 1 and 2 have strategies, $X_1=\{a_1, a_2\}$, and $X_2=\{b_1, b_2\}$, and the production function $f$ is given by: $f(a_1,b_1)=0$, $f(a_1,b_2)=1$, $f(a_2,b_1)=2$, and $f(a_2,b_2)=3$. Agents' payoffs  are described in Table \ref{Table3}. The environment $\mathcal{E}^f$ describes a strictly monotonic economy, but it is unfair. Similarly, by replacing the production function $f$ by another function $g$ defined by: $g(a_1,b_1)=0$, $g(a_1,b_2)=g(a_2,b_1)=1$, and $g(a_2,b_2)=3$, we obtain a weakly free monotonic and unfair economy $\mathcal{E}^g$ with agents' payoffs described in Table \ref{Table3bis}.

\begin{table}[htb]\hspace*{\fill}
\begin{game}{2}{2}[Agent~1][Agent~2]
& $b_1$ & $b_2$ \\
$a_1$ &$(0,0)$ &$(2,-1)$ \\
$a_2$ &$(2,-1)$ &$(1,2)$
\end{game}\hspace*{\fill}
\caption{A 2-agent weakly monotonic free and unfair economy}\label{Table3bis}
\end{table}

Note also that neither strategic form game  $G^{\mathcal{E}^f}$ described in Table \ref{Table3}, nor  $G^{\mathcal{E}^g}$ described in Table \ref{Table3bis} admit a pure strategy Nash equilibrium. This shows that the monotonicity conditions do not guarantee the existence of a pure strategy Nash in a free economy that is unfair; and even when an equilibrium exists in such an economy, it may be Pareto-inefficient. This latter situation occurs, for example, in the prisoner's dilemma game. An economy that is represented by a prisoner's dilemma game is monotonic, but its unique equilibrium is Pareto-inefficient (see, for instance, the game described in Table \ref{PDG}; the unique pure strategy Nash equilibrium (Defect, Defect) is Pareto-inefficient).  
\begin{table}[htb]\hspace*{\fill}
\begin{game}{2}{2}[Agent~1][Agent~2]
& Cooperate & Defect \\
Cooperate &$(0, 0)$ &$(-2, 1)$ \\
Defect &$(1, -2)$ &$(-1, -1)$
\end{game}\hspace*{\fill}
\caption{A prisoner's dilemma game}\label{PDG}
\end{table}



\section{A free economy with social justice and inclusion}\label{sec:socialjustice}

\ \ \ Our conception of a free economy with social justice embodies both the ideals of market justice and social inclusion. Members of a society do not generally have the same abilities. Consequently, distribution schemes that are based on market justice alone will penalize individuals with less opportunities or those who are unable to develop a positive productivity to the economy.

One of the goals of social justice is to remedy this social disadvantage that results mainly from arbitrary factors in the sense of moral thought. Social justice requires caring for the least well-off and those who have natural limitations not allowing them to achieve as much as they would like to. This requirement goes beyond the considerations of a free and fair economy in which agents have equal access to civic rights, wealth, opportunities, and privileges. The ideal of social justice could be implemented in a fair society through specific redistribution rules, and that is the main message that we intend to provide in this section. 

Market justice as defined in the previous sections requires that the collective outcome must be distributed based on individual marginal contributions. Thus, a citizen who is not able to contribute a positive value to the economy shouldn't receive a positive payoff. 

Social justice differs to market justice in the sense that everyone should receive a basic worth for living. This principle is consistent with the results found by \citet{de2013fairness} in a recent experimental study in which neutral agents (called ``Decision Makers") are called upon to distribute collective rewards among other agents (called ``Recipients"). They show that even if collective rewards depend on complementarity and substitutability between recipients, some decision markers still allocate  positive rewards to those who bring nothing to the economy. Moreover, a linear convex combination of the Shapley value \citep{shapley1953value} and the \textit{equal split} scheme  arises as a one-parameter allocation estimate of data. This convex allocation is also known as an \textit{egalitarian Shapley value} \citep{joosten1996dynamics}. Intuitively, this pay scheme can be viewed as implementing a progressive redistribution policy where a positive amount of the total surplus in an economy is taxed and redistributed equally among all the agents. We use this distribution scheme to showcase our purpose. We will see that some properties of an economy that embeds the idea of social justice depends on the tax rate. Below, we define the \textit{equal-split}, and an \textit{egalitarian Shapley value} schemes. 

\begin{definition}\label{equal-egalitarianrules}
Let $\mathcal{E}=(N, X,o, f, \phi, u)$ be a free economy.
\begin{enumerate}
    \item $\phi$ is the \textit{equal split} distribution scheme, if
    \begin{equation*}
       \phi_i(f,x)=\frac{f(x)}{n}, \ \text{for all} \ f\in P(X), x\in X, \ \text{and} \ i \in N.
    \end{equation*}
    
    \item $\phi$ is an \textit{egalitarian Shapley value} if there exists $\alpha \in [0,1]$ such that for all $f\in P(X)$, and  $i \in N$,
\begin{equation*}
    \phi_i(f,x)=\alpha \cdot \vect{Sh}_i(f,x)+(1-\alpha)\cdot \frac{f(x)}{n}, \ \text{for all} \ x \in X.
\end{equation*}
\end{enumerate}
\end{definition}
We denote by $\vect{ES}^\alpha$ the egalitarian Shapley value associated to a given $\alpha\in [0,1]$. The mixing equal split and Shapley value satisfies the principles of anonymity and local efficiency, but violates marginality and unproductivity when $\alpha \in [0, 1)$. The allocation scheme  $\vect{ES}^\alpha$ has a very natural interpretation.  Given an outcome $x$, the technology $f$ produces the output $f(x)$. A share ($\alpha$) of the latter is shared among agents according to their marginal contributions, while the remaining ($1-\alpha$) is shared equally among the entire population; the fraction $1-\alpha$ is the tax rate. Immediately, those who are more talented will still receive more under a given egalitarian Shapley value scheme, but less compared to what they receive in a free and fair economy (when $\alpha = 1$). Additionally, those who do not have the opportunity to contribute to their optimum scale will still be rewarded. We have the following definition.

\begin{definition}\label{equal-egalitarianrules}
$\mathcal{E}=(N, X,o, f, \phi, u)$ is a \textit{free economy with social justice} if there exists $\alpha \in [0,1[$ such that $\phi = \vect{ES}^{\alpha}$. We call $\mathcal{E}^{\alpha}=(N, X,o, f, \vect{ES}^{\alpha}, u)$ an $\alpha$-free economy with social justice. 
\end{definition}

In Section \ref{sec:efficiencysocialjustice}, we analyze equilibrium existence and Pareto-efficiency in free economies with social justice. Our methodology is similar to the one followed in Sections \ref{sec:faireconomyandequilibrium} and \ref{sec:efficiency}. In Section \ref{sec:referencepoint}, we prove that an economy can always choose its reference point to induce equilibrium efficiency, even when the economy is not monotonic.

\subsection{Equilibrium existence and efficiency in a free economy with social justice}\label{sec:efficiencysocialjustice}

In what follows, we study the existence of equilibrium in  an $\alpha$-free economy with social justice. As defined in Section \ref{fair-strategygames}, a free economy with social justice admits an equilibrium if the strategic  form game derived from that economy possesses a pure strategy Nash equilibrium. A meritocratic planner will choose a higher $\alpha$ when allocating resources since talents and merits have more value in such a society.  An egalitarian planner will put a higher weight on equal distribution. It follows that a choice of $\alpha$ reveals a trade-off between market justice and egalitarianism. The good news is that there exists a self-enforcing social contract irrespective of the size of $\alpha$. We have the result hereunder. 

\begin{theorem}
\label{theo:Egalitarian Shapley value}
Any $\alpha$-free economy with social justice $\mathcal{E}^{\alpha}=(N, X, o,f, \vect{ES}^\alpha, u)$ admits an equilibrium.
\end{theorem}

\begin{proof}[Proof of Theorem \ref{theo:Egalitarian Shapley value}]
Consider $\alpha \in [0,1]$ such that $\phi = ES^{\alpha}$. In the proof of Theorem \ref{theorem: non emptyness: shapley}, we show that the sum of excess payoffs in any cycle of deviations from any strategic  form game derived from a fair economy equals $0$. The same result holds for any strategic form game derived from an $\alpha$-free economy with social justice, since an egalitarian Shapley value is a linear combination of the Shapley value and equal division. Thus, we conclude the proof.
\end{proof}

We also provide a condition under which a free economy with social justice has a Pareto-efficient economy. We have the following definition. 
\begin{definition}
Let $\mathcal{E}^{\alpha}=(N, X, o,f, \vect{ES}^{\alpha}, u)$  be an $\alpha$-free economy with social justice. An \textit{optimal outcome} is any outcome $x\in \arg \max\limits_{y \in X}f(y)$ at which $f$ is maximized. 
\end{definition}

The following result is deduced from Theorem \ref{monotonicresult2}. 

\begin{cor} \label{monotonicresultsocialjustice}
A weakly monotonic $\alpha$-free economy with social justice $\mathcal{E}^{\alpha}=(N, X, o,f, \vect{ES}^{\alpha}, u)$  admits an equilibrium that is Pareto-efficient. If $f$ is strictly monotonic, then, the equilibrium is unique and Pareto-efficient. 
\end{cor}

The proof of Corollary \ref{monotonicresultsocialjustice} is similar to that of Theorem \ref{monotonicresult2}. Next, we provide an additional result about Pareto-efficiency of equilibria in a free economy with social justice.

\begin{theorem}\label{theo:Egalitarian Shapley value efficiency}
There exists $\alpha_0 \in (0,1)$ such that for all $\alpha \in [0,\alpha_0]$, the $\alpha$-free economy with social justice $\mathcal{E}^{\alpha}=(N, X, o,f, \vect{ES}^\alpha, u)$ admits an  equilibrium that is Pareto-efficient. 
\end{theorem}

\begin{proof}[Proof of Theorem \ref{theo:Egalitarian Shapley value efficiency}]
Assume that $\alpha$ is sufficiently  small. If $f$ admits a unique optimal outcome $x$, then $x$ is a pure strategy Nash equilibrium of the game generated by any $\alpha$-free economy with social justice $\mathcal{E}^{\alpha}$. In the case $f$ admits two or more optimal outcomes, then, for strictly positive but sufficiently small $\alpha$, no agent has any incentive to deviate from an optimal outcome to a non-optimal outcome. As games generated by $\alpha$-free economies with social justice do not admit cycles of deviations, it is not possible to construct any cycle of deviations within the set of optimal outcomes. It follows that at least one optimal outcome is a Nash equilibrium. The latter profile is also Pareto-efficient as it maximizes the sum of agents' payoffs. 
\end{proof}

\begin{example}[Taxation and Social Justice]
Consider a small economy involving three agents, $N=\{1, 2, 3\}$, who live in three different states or regions in a given country. One can assume that each agent is the ``typical" representative of each state. Agents face different occupational choices. Agent 1 can decide to stay unemployed (strategy ``$a$"), work in a middle class job (strategy ``$b$") that provides an annual salary of \$188,000, or accumulate experience to land a higher skilled job (strategy ``$c$") that pays an annual salary of \$200,000. Agent 2 can only choose between strategies ``$a$" and ``$b$". For many reasons including health concerns, natural disasters such as hurricane, pandemics or wildfire, or civil war violence, agent 3 does not have the opportunities available to other agents; he or she can not work, and is therefore considered as unemployed. The government uses marginal tax rates to determine the amount of income tax that each agent must pay to the tax collector. The aggregate annual fiscal revenue function $f$ for the economy depends on agents' strategies and it is described as follows: $f(a,a,a)=0$, $f(a,b,a)= \$41,175.5$, $f(b,a,a)= \$41,175.5$, $f(b,b,a)=\$82,351$, $f(c,a,a)=\$45,015.5$, and $f(c,b,a)=\$86,191$. Numerous countries over the world use marginal tax brackets to collect income taxes (see, for example, a report by \citet{Bunn2019} for the Organisation for Economic Co-operation (OECD) and the Development and European Union (EU) countries). The function $f$ is a simplified version of such fiscal revenue rules. With the tax revenue collected, the government provides public goods.  However, the type of public investment received by an agent's state depends on the agent's marginal contribution to the aggregate annual fiscal revenue. Using the Shapley scheme $\phi = \vect{Sh}$ in the distribution of public investments yields the outcome $x^{*} = (c, b, a)$ as the unique pure strategy Nash equilibrium in this free and fair economy. At this equilibrium, the state of agent 1 receives a public good that is worth \$45,015.5, agent 2's state receives a public investment of \$41.175.5, and agent 3's state receives nothing. However, if the egalitarian Shapley scheme $\phi = \vect{ES}^{4/5}$ is used instead to redistribute the fiscal revenue, then $x^{*} = (c, b, a)$ is still the unique pure strategy Nash equilibrium in the free economy with social justice. In that case, the outcome $x^{*}$ is still Pareto-efficient and the ranking of the size of investment across states does not change. Agent 3's state receives a public investment of \$5,746, agent 2's state receives \$38,686.5, and agent 1's state receives \$41,758.5.  Although the allocation $\vect{ES}^{4/5} (f, x^{*}) = (\$41,758.5, \$38,686.5, \$5,746)$ might not be the ``best" decision for some people living in that society, it is a significant improvement (at least for agent 3's state) from the market allocation $\vect{Sh}(f, x^{*})=(\$ 45,015.5, \$ 41,175.5, 0)$. 
\end{example}

Using Theorem \ref{theo:Egalitarian Shapley value efficiency}, we deduce the following corollary. 

\begin{cor}\label{prop:positiverewards}
Let $\mathcal{E}^{\alpha}=(N, X, o,f, \vect{ES}^\alpha, u)$  be an $\alpha$-free economy with social justice. Assume that $f$ only takes non-negative values. Then, each agent receives a non-negative payoff at any equilibrium.  
\end{cor} 

The intuition behind Corollary \ref{prop:positiverewards} is straightforward. Assuming that at a given outcome $x \in X$, $f(x)$ is non-negative, then for all $i\in N$, agent $i$'s payoff is non-negative if instead of choosing $x_i$, the agent chooses the reference point $o_i$. 

\subsection{Choosing a reference point to achieve equilibrium efficiency}\label{sec:referencepoint} 

So far, we have assumed that the reference point $o$ is exogenously determined and that in a free economy,  the surplus function $f$ is such that $f(o_1, o_2, ..., o_n)=0$. As noted earlier, this latter point is just a simplifying normalization. We have also shown that in a free and fair economy, all the equilibria may be Pareto-inefficient, especially in the absence of monotonicity. Similarly, in a free economy with social justice, if the tax rate ($1-\alpha$) is too small, a Pareto-efficient equilibrium may not exist either. This section shows that we can achieve equilibrium efficiency simply by changing the reference point of any free and fair economy or any free economy with social justice.

Without loss of generality, we assume that  $f(o)$ is strictly positive and modify the Shapley distribution scheme such that for $i\in N$, and $x\in X$, agent $i$'s payoff at $(f, x)$, denoted $\overline{\vect{Sh}} (f,x )$, is given by  $\overline{\vect{Sh}} (f,x ) = \vect{Sh}_i(f-f(o),x)+\frac{f(o)}{n}$. Let us denote $\overline{P}(X)= \{g:X\rightarrow \mathbb{R}, \ \text{with} \ g(o)>0\}$. Our next result says that any optimal outcome can be achieved via an equilibrium profile in any $\alpha$-free economy with social justice endowed with the distribution scheme $\overline{\vect{ES}}^{\alpha}$, where $\overline{\vect{ES}}^{\alpha} (f,x)=\alpha \cdot \overline{\vect{Sh}}_i(f,x)+(1-\alpha)\cdot \frac{f(x)}{n}, \ \text{for all} \ x \in X$ and $f\in \overline{P}(X)$.

\begin{theorem}\label{res:referencepoint}
For all free economy $\mathcal{E}^{\alpha} (o)=(N,X, o, f, \overline{\vect{ES}}^{\alpha}, u)$, there exists another reference outcome $o'$ such that the $\alpha$-free economy $\mathcal{E}^{\alpha}(o')=(N,X, o', f, \overline{\vect{ES}}^{\alpha}, u)$ admits an optimal equilibrium $x^{*}$.
\end{theorem}

\begin{proof}[Proof of Theorem \ref{res:referencepoint}]
Assume $\alpha=1$. Let $o^\prime$ be a profile of inputs such that  $f(o^\prime)= \max \limits_{x\in X}f(x)$. No agent has any strict incentive to deviate from $o'$. Indeed if agent $i$ deviates and chooses $x_i$, then agent $i$ is the only active agent at the new outcome $(o^\prime_{-i}, x_i)$. As each inactive agent receives $\frac{f(o^\prime)}{n}$ at $(o^\prime_{-i},x_i)$, and $f(o^\prime)$ maximizes the production, it follows from  the local efficiency axiom of the Shapley distribution scheme that the deviation $x_i$ is not strictly profitable. A similar argument holds for any other $\alpha \in [0,1)$. Indeed, at the profile $(o^\prime_{-i},x_i)$, agent $i$ receives $\alpha\left(f(o^\prime_{-i},x_i)-f(o^\prime)+\frac{f(o^\prime)}{n}\right)+(1-\alpha)\frac{f(o^\prime_{-i},x_i)}{n}$, which is less than $\frac{f(o^\prime)}{n}$.
\end{proof}

Remark that this result holds for any value of $\alpha$, including for $\alpha=1$, which corresponds to a situation where the tax rate is zero. In that case, the entire surplus of the economy is distributed following the Shapley value. The analysis implies that if an economy can choose its reference point, it can always do so to lead to equilibrium efficiency.   

\section{Some applications}\label{sec:application}

There a wide variety of applications of our theory. In this section, we provide applications to the distribution of surplus in a firm, exchange economies, self-enforcing lockdowns in a networked economy facing a pandemic, and publication bias in the academic peer-review system.

\subsection{Teamwork: surplus distribution in a firm}
\ \ \ In this
first application, we use our theory to show how bonuses can be distributed among workers in a way that incentivizes them to work efficiently.

Consider a firm which consists of a finite set of workers $N=\{1, 2, ..., n\}$. Each worker $i \in N$ privately and freely chooses an effort level $x^j_i\in X_i$, and bears a corresponding non-negative cost $c^j_i = c(x^j_i)$, where $c(.)$ denotes the cost function. The cost of labor supply includes any private  resources or extra working time that worker $i$ puts into the project (for example, transportation costs, time, etcetera). Workers' labor supply choices are made simultaneously and independently. The owner of the firm (or the team leader) knows the cost associated to each effort level. At each effort profile $x=(x_1,\cdots,x_n)$, a corresponding monetary output $F(x)$ is produced. A fraction of the monetary output, $f=\gamma \cdot F$, with $\gamma \in (0, 1)$, is redistributed to workers in terms of bonuses. 

\vspace{2mm}

The existence of a pure strategy Nash equilibrium in this teamwork game follows from Lemma \ref{lemma:deviation cycle}. To see this, observe that the payoff function of a worker can be decomposed in two parts: the bonus that is determined by the Shapley payoff and the cost function. Lemma \ref{lemma:deviation cycle} shows that the sum of excess payoffs in any cycle of deviations equals 0 in any free and fair economy (or any strategic game with Shapley payoffs). The reader can check that the sum of excess costs in any cycle of strategy profiles is zero as well in the game. The latter implies that the sum of excess payoffs in any cycle of strategy profiles of the teamwork game is equal to 0. Therefore, the teamwork game admits no cycle of deviations. As the game is finite, we conclude that it admits at least a Nash equilibrium profile in pure strategies. (Recall that the total output of the firm, $F$, and the total bonus, $f$, are perfectly correlated.) We should point out that a pure strategy Nash equilibrium always exists in the teamwork game, even if costs are high. In the latter case, some workers, if not all, might find it optimal to remain inactive at the equilibrium. In such a situation, the owner might want to raise the total bonus to be redistributed to workers. 
\vspace{1mm}

\textit{Illustration}. We now provide a numerical example with two workers called Bettina and Diana. Bettina has four possible effort levels: $b_1,\ b_2,\ b_3$ and $b_4$; and Diana has four possible effort levels as well: $d_1,\ d_2, \ d_3$ and $d_4$. The cost functions of the two workers are given by: $c(b_1)=c(d_1)=0,\ c(b_2)=c(b_3)=c(d_2)=c(d_3)=4$, $c(b_4)=3$, and $c(d_4)=5$. The fraction  $f$ of the output to be redistributed as bonus is described in  Table \ref{table: Total bonus}. The number $f(b,d)$ is the bonus to be distributed at the profile of efforts $(b,d)$; for instance, $f(b_1, d_1)= 0$. \\

\begin{table}[h]\hspace*{\fill}%
\begin{game}{4}{4}[Bettina][Diana]
& $d_1$ & $d_2$ & $d_3$ & $d_4$\\
$b_1$ &$0$ &$5$ &$1$ & $13$\\
$b_2$ &$2$ &$8$ & $10$ & $2$\\
$b_3$ &$5$ &$13$ &$1$ & $13$\\
$b_4$ &$3$ &$9$ & $13$ & $2$\\
\end{game}\hspace*{\fill}%
\caption{Total bonus function in a teamwork game}
\label{table: Total bonus}
\end{table}
The corresponding Shapley payoffs are described in Table \ref{table: Shapley payoffs} and the net payoffs of Bettina and Diana in the teamwork game are described in Table \ref{table: Teamwork game}.

\begin{table}[h!]\hspace*{\fill}%
\begin{game}{4}{4}[Bettina][Diana]
& $d_1$ & $d_2$ & $d_3$ & $d_4$\\
$b_1$ &$(0,0)$ &$(0,5)$ &$(0,1)$ & $(0,13)$\\
$b_2$ &$(2,0)$ &$(\frac{5}{2},\frac{11}{2})$ & $(\frac{11}{2},\frac{9}{2})$ & $(-\frac{9}{2},\frac{13}{2})$\\
$b_3$ &$(5,0)$ &$(\frac{13}{2},\frac{13}{2})$ &$(\frac{5}{2},-\frac{3}{2})$ & $(\frac{5}{2},\frac{21}{2})$\\
$b_4$ &$(3,0)$ &$(\frac{7}{2},\frac{11}{2})$ & $(\frac{15}{2},\frac{11}{2})$ & $(-4,6)$\\
\end{game}\hspace*{\fill}%
\caption{Shapley payoffs: redistribution of total bonus in a teamwork game}
\label{table: Shapley payoffs}
\end{table}

\begin{table}[h!]\hspace*{\fill}%
\begin{game}{4}{4}[Bettina][Diana]
& $d_1$ & $d_2$ & $d_3$ & $d_4$\\
$b_1$ &$(0,0)$ &$(0,1)$ &$(0,-3)$ & $(0,8)$\\
$b_2$ &$(-2,0)$ &$(-\frac{3}{2},\frac{3}{2})$ & $(\frac{3}{2},\frac{1}{2})$ & $(-\frac{17}{2},\frac{3}{2})$\\
$b_3$ &$(1,0)$ &$(\frac{5}{2},\frac{5}{2})$ &$(-\frac{3}{2},-\frac{11}{2})$ & $(-\frac{3}{2},\frac{11}{2})$\\
$b_4$ &$(0,0)$ &$(\frac{1}{2},\frac{3}{2})$ & $(\frac{9}{2},\frac{3}{2})$ & $(-7,1)$\\
\end{game}\hspace*{\fill}%
\caption{Bettina and Diana's net payoffs in a teamwork game}
\label{table: Teamwork game}
\end{table}

The profile $(b_4,d_3)$ is a pure strategy Nash equilibrium. Therefore, the owner of the firm can implement the profile \textbf{$(b_4,d_3)$} without any need of monitoring the actions of Bettina and Diana, as $(b_4,d_3)$ is self-enforcing. The owner can implement the profile \textbf{$(b_1,d_4)$} as well. Note that the set of equilibrium effort profiles depend on the cost functions, and that no worker receives a non positive bonus at the equilibrium. The reason is that each worker $i$ always has the option to remain inactive, which is equivalent to Bettina choosing $b_1$ or Diana choosing $d_1$ in this illustration. The two equilibria in this teamwork game are Pareto-efficient.

\subsection{Contagion and self-enforcing lockdown in a networked economy} 

\ \ \ In this section, we provide an application of a free and fair economy to contagion and self-enforcing lockdown in a networked economy. We show how the costs of a pandemic from a virus outbreak can affect agents' decisions to form and sever bilateral relationships in the economy. Specifically, we illustrate this application by using the contagion potential of a network \citep{pongou2010economics, pongouserrano2013, pongou2016volume, PongouTondji2018}.

\vspace{2mm}

Consider an economy $\mathcal{M}$ involving agents who freely form and sever bilateral links according to their preferences. Agents' choices lead to a network, defined as a set of bilateral links. Assume that rational behavior is captured by a certain equilibrium notion (for example, Nash equilibrium, pairwise-Nash equilibrium, etc.). Such an economy may have multiple equilibria. Denote by $\mathcal{E(M)}$ the set of its equilibria. Our main goal is to assess agent's decisions in response to the spread of a random infection (for example, COVID-19) that might hit the economy. As the pandemic evolves in the economy, will some agents decide to sever existing links and self-isolate themselves? How does network structure depend on the infection cost? 

\vspace{2mm}

To illustrate these concepts, consider an economy involving a finite set of agents $N=\{1, ..., n\}$. All agents  simultaneously announce the direct links they wish to form. For every agent $i$, the set of strategies is an $n$-tuple of 0 and 1, $X_{i} = \left\{0, 1\right\}^n$. Let $x_i= (x_{i1}, ..., x_{ii-1}, 1, x_{ii+1},..., x_{in})$ be an element in $X_i$. Let $x_{ij}$ denote the $j$th coordinate of $x_i$. Then, $x_{ij}=1$ if and only if $i$ chooses a direct link with $j \ (j \neq i)$, or $j=i$ (and thus $x_{ij} = 0$, otherwise).  We assume that the formation of a link requires mutual consent, that is, a link $ij$ is formed in a network if and only if $x_{ij} x_{ji} = 1$. We denote $X=\times_{j\in N} X_j$. An outcome $x \in X$ yields a unique network $g(x)$. However, a network can be formed from multiple outcomes. We denote $o= (0, .., 0)$ the reference outcome, and $g(o)$ the empty network. It follows that the networked economy $\mathcal{M}$ can be represented by a free economy $(N, X, o, f, \phi,u)$, where $f$ is the production function and $u=\phi$ the payoff function (see below).

\vspace{2mm}

Assume that rationality is captured by the notion of pairwise-Nash equilibrium as defined by, among others, \citet{calvo2004job}, \citet{goyal2006unequal}, and \citet{bloch2007formation}. The concept of pairwise-Nash equilibrium refines Nash equilibrium building upon the pairwise stability concept in \citet{jackson1996strategic}. Pairwise-equilibrium networks are such that no agent gains by reshaping the current configuration of links, neither by adding a new link nor by severing any subset of the existing links. Let $g$ be a network and $ij \in g$ a link. We let $g+ij$ denote the network found by adding the link $ij$ to $g$, and $g-ij$ denote the network obtained by deleting the link $ij$ from $g$. Formally, $g$ is a pairwise-Nash equilibrium network if and only if there exists a Nash equilibrium outcome $x^{*}$ that supports $g$, that is $g=g(x^{*})$, and for all $ij \notin g$, $\phi_i(f, g +ij) > \phi_i(f, g)$ implies $\phi_j(f, g +ij) < \phi_j(f, g)$.

\vspace{2mm}

The contagion function is the contagion potential of a network \citep{pongou2010economics, pongouserrano2013, pongou2016volume, PongouTondji2018}. To define this function, we consider a network $g$ that has $k$ components, where a component is a maximal set of agents who are directly or indirectly connected in $g$; and $n_{j}$ the number of individuals in the $j^{th}$ component $(1\leq j\leq k)$. \citet{pongou2010economics} shows that if a random agent is infected with a virus, and if that agent infects his or her partners who also infect their other partners and so on, the fraction of infected agents is given by the contagion potential of $g$, which is: 
\begin{equation*}
\mathcal{P}(g)=\frac{1}{n^{2}}\sum\limits_{j=1}^{k}n_{j}^{2}.  
\end{equation*}
However, in a network $g$, each agent is exogenously infected with probability $\frac{1}{n}$, and given that agents are not responsible for exogenous infections, the part of contagion for which agents are collectively responsible in $g$ is: 
\begin{equation*}
 \Tilde{c}(g)= \mathcal{P}(g)-\frac{1}{n}.   
\end{equation*}
We assume that the infection by a communicable virus leads to a disease outbreak in the economy. Measures that are implemented to fight the pandemic bring economic costs to society. To assess those costs, we assume that the collective contagion function $\Tilde{c}$ generates a pandemic cost function $\mathcal{C}$ so that, for each network $g$:
\begin{equation*}
\mathcal{C}(g)= F(\Tilde{c}(g)), \ F \ \text{being a well-defined function}. 
\end{equation*}
The pandemic and network formation affect economic activities. The formation of a network $g$ brings an economic value $v(g) \in \mathbb{R}$ to the economy. Given the cost function $\mathcal{C}$, the economic surplus of a network $g$ is:
\begin{equation*}
f(g)= v(g)- \mathcal{C}(g).
\end{equation*}
Our main goal is to examine each agent's behavior in forming or severing bilateral links as the pandemic spreads in the economy. Let $g$ be a network and $S$ be a set of agents. We denote by $g^S$ the restriction of the network $g$ to $S$. This restriction is obtained by severing all the links involving agents in $N \backslash S$. Also, let $i$ be an agent. We denote by $g^S+i$ the network $g^{S\cup \left\{i\right\}}$ obtained from $g^S$ by connecting $i$ to all the agents in $S$ to whom $i$ is connected in the network $g$. The structure of the networked economy provides a natural setting for the use of the Shapley distribution scheme. In a competitive environment where marginal contributions are the only inputs that matter in the economy, we can expect that an agent who adds no value to any network configuration receives no payoff, and a more productive agent in a network structure receives a payoff that is greater relative to that of less productive agents. Assuming that the output from individual contributions are entirely shared among agents, it becomes natural to consider that agent $i$'s payoff in a network $g$ is given by the Shapley distribution scheme (\ref{Shapleyfunction}):
\begin{equation*}
\phi_i(f, g) \equiv \vect{Sh}_i(f, g) = \sum \limits_{S \subseteq N, \ i\neq S} \frac{s! (n-s-1)!}{n!} \ \left\{ f (g^S +i) - f(g^S) \right\}, \ s= |S|.
\end{equation*}

\vspace{2mm}

The networked economy $\mathcal{M}= (N, X, o, f, \vect{Sh},u)$ describes a free and fair economy. We have the following result. 

\begin{proposition}
Pairwise-Nash equilibrium networks always exist: $\mathcal{E(M)} \neq \emptyset$. 
\end{proposition}

This result partly follows from Theorem \ref{theorem: non emptyness: shapley}, but is stronger because the notion of pairwise-Nash equilibrium refines the Nash equilibrium. The proof is left to the reader. We illustrate it below.

\textit{Illustration}. Let $N=\{1,2,3\}$. Assume the set of an agent $i$’s direct links in a network $g$ is $L_i(g)= \left\{jk \in g: j=i \ \text{or} \ k=i, \ \text{and} \ j\neq k\right\}$, of size $l_i(g)$. The size of $g$ is $l(g)=\sum \limits_{i\in N} l_i(g)/2$. Note that $l(g)=0$ if and only if $g$ is the empty network. For illustration, we assume that for each  network $g$: 
\begin{equation*} \label{eq1}
\begin{split}
v(g) & = [l(g)]^{1/2} \\
\mathcal{C}(g) & = \lambda \Tilde{c} (g) = \lambda [\mathcal{P}(g)-\frac{1}{n}], \ \lambda>0 \\
f(g) & = [l(g)]^{1/2}- \lambda [\mathcal{P}(g)-\frac{1}{n}], \lambda>0.
\end{split}
\end{equation*}
We can rewrite $f$ as follows (note that $\mathcal{P}(\emptyset) = \frac{1}{n})$: 
\begin{align*} 
 f(g) &= 
 \begin{cases} 
   0 & \text{if } l(g) = 0 \\
   1-\frac{2 \lambda}{9} & \text{if } l(g) =1\\
    \sqrt{2}-\frac{2 \lambda}{3} & \text{if } l(g) =2\\
    \sqrt{3}-\frac{2 \lambda}{3} & \text{if } l(g) =3
  \end{cases} 
\end{align*}
Given that there is only three agents, we can fully represent the set of networks in $\mathcal{M}$. The agents are labeled as described in Figure \ref{agentlabel}. 
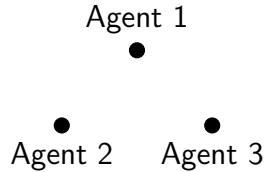
\begin{figure}[!h]
\centering
\begin{tikzpicture}[scale=1]
    \filldraw[color=black] (0,0) circle (0.1) ;
    \draw (0,-0.08) node[below]{Agent 2};
    \filldraw[color=black] (2,0) circle (0.1);
    \draw (2,-0.08) node[below]{Agent 3};
   \filldraw[color=black] (1,1) circle (0.1);
    \draw (1, 1.08) node[above]{Agent 1}; 
    \end{tikzpicture}
    \caption{Disposition of agents in a network} \label{agentlabel}
\end{figure}
In Figure \ref{networkset}, we display the different network configurations in $\mathcal{M}$. In each network, the payoff of each agent is given next to  the corresponding node. The pairwise stability concept facilitates the search of equilibrium networks. We have the following result. We denote by $g^N$ the complete network.
\begin{proposition}\label{eqnetwork}
Let $g$ be a network. If:
\begin{enumerate}
    \item $\lambda < 1.8 \sqrt{2}-0.9$, then $\mathcal{E(M)} = \{g^N\}$.
    \item $1.8 \sqrt{2}-0.9 < \lambda < \frac{3\sqrt{3}}{2}$, then $g \in \mathcal{E(M)}$ if and only if $l(g) \in \{1, 3\}$. 
    \item $\frac{3\sqrt{3}}{2} < \lambda < 4.5$, then $g \in \mathcal{E(M)}$ if and only if $l(g) =1$.
    \item $\lambda > 4.5$, then $\mathcal{E(M)} = \{g(o)\}$.
\end{enumerate}
\end{proposition}

\begin{figure}[h] 
\centering
\includegraphics[width=0.7\textwidth]{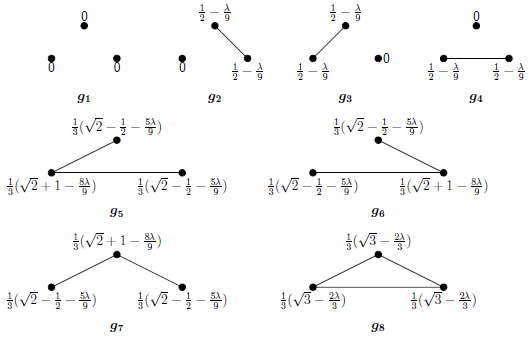}
 \caption{Possible network formation in $\mathcal{M}$}\label{networkset}
\end{figure}
The proof of Proposition \ref{eqnetwork} is straightforward and left to the reader. Clearly, Proposition \ref{eqnetwork} shows that pandemic costs affect agents' decisions in the networked economy. The parameter $\lambda$ summarizes the negative effects of the contagion in the economy. When there is no disease outbreak, or the pandemic costs are very low (lower values of $\lambda$), each agent gains by keeping bilateral relationships with others. In that situation, the complete network is likely to sustain as the equilibrium social structure in the economy. No agent has an incentive to self-isolate. However, as the pandemic costs rise, agents respond by severing some bilateral connections. For intermediate values of $\lambda$ ($\frac{3\sqrt{3}}{2} < \lambda < 4.5$), only networks with one link will be sustained in the equilibrium. This means that some agents find it rational to partially or fully self-isolate in order to reduce the spread of the virus. In the extreme case where the contagion costs are very high ($\lambda > 4.5$), a complete lockdown arises, and the empty network is the only equilibrium. 

Interestingly, the value of $\lambda$ depends on the nature of the virus. Viruses induce different severity levels. For example, COVID-19 and the flu virus have different values, inducing different network configurations in equilibrium. The different network configurations in Figure \ref{networkset} can therefore be interpreted as the networks that will arise in different scenarios regarding the nature of the virus.   

\subsection{Bias in academic publishing}\label{sec:academicreview}

\ \ \ In this section, we apply the model of a free and fair economy  to academic publishing in a knowledge environment. Generally, academic researchers have freedom to choose research topics that are likely to be published either in peer-reviewed or non-peer-reviewed outlets. However, studies show that the peer-review process is not generally anonymous, and it involves some biases (see, for example, \citet{ellison2002evolving}, \citet{heckman2017publishing}, \citet{serrano2018top5itis}, \citet{akerlof2020sins}, and the references therein). Following \citet{ellison2002evolving}, we consider a model of producing scientific knowledge in which researchers differentiate topics along two quality dimensions: importance (or $q$-quality) and hardness (or $r$-quality).\footnote{Tough the trade-off between the two quality dimensions can be viewed as a rational decision, the consequences can be detrimental to economics, as a discipline and profession. For instance, some general interest journals suffer from the ``incest factor" \citep{heckman2017publishing}, and \citet{akerlof2020sins} shows that the tendency of rewarding ``hard" topics versus ``soft" topics in economics results in ``sins of omissions" where issues that are relevant to the literature and can not be approached in a ``hard" way are ignored.} In the hypothetical and straightforward knowledge economy that we analyze, we assume that both importance and hardness levels are discrete, ordered, and are homogeneous among researchers. Formally, $Q= \{q_0, q_1, ..., q_m\}$ denotes the set of importance levels, with $q_0<q_1<...<q_m$, and $R= \{r_0, r_1, ..., r_m\}$ denotes the set of different degrees of hardness, with $r_0 <r_1<...<r_m$.

We consider a knowledge environment involving a finite set of researchers $N=\{1,..., n\}$. Each researcher selects a topic of a given importance level and degree of hardness. For every researcher $i$, a strategy $x_i = (q^i, r^{i}) \in X_i \subseteq Q \times R$, where $q^i \in Q$, and $r^i \in R$. We denote $X=\times_{j \in N} X_j$. We consider $o_i = (q_0, r_0)$ as the reference choice for researcher $i$, and $o= (o_j)_{j\in N}$ the reference outcome. A knowledge function (or technology) $f$ transforms any outcome $x \in X$ to the number of published articles $f(x) \in \mathbb{R}$, with $f(o)=0$. An allocation $\phi$ distributes $f(x)$ to active researchers so that the utility of researcher $i$, $u_i$, at the profile $x$ given the knowledge function $f$, is $u_i(x)= \phi_i(f, x)$. 

The knowledge economy $\mathcal{E}^{\phi}= (N, X, o, f, \phi, u)$ defines a free economy. Thanks to Theorem \ref{theorem: non emptyness: shapley}, the free and fair knowledge economy $\mathcal{E}^{\vect{Sh}} = (N, X, o, f, \vect{Sh}, u)$ admits a pure strategy Nash equilibrium. The allocation $\phi$ ($\neq \vect{Sh}$) in the free knowledge economy can be viewed as the current academic publishing system. As mentioned above, the latter could lead to an equilibrium outcome that shows a bias towards ``hardness" and against ``importance". To illustrate our point, we  consider a simple knowledge economy involving two active researchers $N= \{1, 2\}$ with the same ``abilities" of producing scholarly articles. Each researcher $i$'s criteria for a choice of topic belongs to the set $X_i = \{0, 1, 2, 3\}$, where each number represents a pair in $Q \times R$: ``0": (soft, less important), ``1": (soft, important), ``2": (hard, less important), and ``3": (hard, important). The knowledge function $f$ matches any profile of decisions $x= (x_1, x_2)$ made by the researchers to $f(x)$, the number of academic articles produced in the economy: $f(0,0)=0$, $f(1,0)= f(0,1)=f(2,2) =10$, $f(1,1)=20$, $f(2,0) = f(0,2)=f(3,1)=f(1,3)=f(2,3)=f(3,2)=8$, $f(3,0)= f(0,3)=4$, $f(1,2)=f(2,1)=14$, and $f(3,3)=6$. We assume that the current academic publishing system allocates articles in the knowledge economy $\mathcal{E}^{\phi}$ according to the allocation scheme $\phi$ described in Table \ref{knowledgeallocation1} below.\footnote{Although we do not have a clear evidence to support the allocation $\phi$, studies such as \citet{heckman2020publishing}, \citet{colussi2018social}, \citet{sarsons2017recognition}, and \citet{card2013nine} document that there exists a preferential treatment for some group of authors in the academic publishing process.} 
\begin{table}[htb]\hspace*{\fill}%
\begin{game}{4}{4}[Researcher~1][Researcher~2]
& $0$ & $1$ & $2$ & $3$\\
$0$ &$(0, 0)$ &$(0, 10)$ &$(0, 8)$ & $(0, 4)$\\
$1$ &$(5, 5)$ &$(5, 15)$ &$(4,10)$ & $(2, 6)$\\
$2$ &$(3, 5)$ &$(6, 8)$ &$(1, 9)$ & $(3, 5)$\\
$3$ &$(3, 1)$ &$(4, 4)$ &$(4, 4)$ & $(2, 4)$
\end{game}\hspace*{\fill}\\
\caption{Academic Knowledge under $\phi$}\label{knowledgeallocation1}
\end{table}
The free economy $\mathcal{E}^{\phi}$ admits a unique equilibrium $x^{*}= (3, 2)$ where both researchers display favor for hardness relative to importance: Researcher 1 favors hard and important, and Researcher 2 favors hard and less important. At that equilibrium $x^{*}$, the economy produces 8 scientific papers. The profile $x^{*}$  is Pareto-dominated by the outcome $(1,1)$ that produces 20 articles in the economy. 

Note that there is another distortion in Table \ref{knowledgeallocation1}. Researcher 1 does not receive the same treatment as Researcher 2. For instance, when Researcher 1 moves from the reference point to the strategy "1", he or she receives the same reward of 5 as Researcher 2. However, when Researcher 2 does the same move, he or she keeps all the benefits, and Researcher 1 receives 0 even if the knowledge function produces the same output at both profiles $(0, 1)$ and $(1,0)$. What would happen in this knowledge economy $\mathcal{E}^{\phi}$ if the Shapley distribution scheme $\vect{Sh}$ replaces $\phi$? 

Well, it is straightforward to show that the researchers are \textit{symmetric} under the knowledge function $f$. Using Anonymity and the other principles of merit-based justice, Table \ref{knowledgeallocation2} below describes the allocation of academic articles under the allocation $\vect{Sh}$. 
\begin{table}[htb]\hspace*{\fill}%
\begin{game}{4}{4}[Researcher~1][Researcher~2]
& $0$ & $1$ & $2$ & $3$\\
$0$ &$(0, 0)$ &$(0, 10)$ &$(0, 8)$ & $(0, 4)$\\
$1$ &$(10, 5)$ &$(10, 10)$ &$(7, 7)$ & $(4, 4)$\\
$2$ &$(8, 0)$ &$(7, 7)$ &$(5, 5)$ & $(4, 4)$\\
$3$ &$(4, 0)$ &$(4, 4)$ &$(4, 4)$ & $(3, 3)$
\end{game}\hspace*{\fill}\\
\caption{Academic Knowledge under $\vect{Sh}$}\label{knowledgeallocation2}
\end{table}

From Table \ref{knowledgeallocation2}, we can easily conclude that the free economy $\mathcal{E}^{\phi}$ is unfair. The identity-bias that we observe under the academic publishing system $\phi$ does not arise in the free and fair knowledge environment because the distribution scheme $\vect{Sh}$ allocates rewards based on marginal contributions. The free and fair knowledge economy $\mathcal{E}^{\vect{Sh}}$ admits the unique profile $x^{**} = (1,1)$ as equilibrium in which both researchers exhibit preferences for soft and important topics.  The outcome $x^{**}$ is Pareto-optimal and it maximizes the quantity of articles produced in the economy. Importantly, researchers produce the same number of articles at the equilibrium given their ``abilities" and the fact that they choose the same strategy. The profile $x^{**}$ in the free and fair economy $\mathcal{E}^{\vect{Sh}}$ strictly dominates the equilibrium outcome in the free knowledge economy $\mathcal{E}^{\phi}$ with the academic publishing system $\phi$.

\subsection{Exchange economies}
In this section, we apply our theory to pure exchange economies (Section \ref{sec:pureexchangeconomy}) and markets with transferable payoff (Section \ref{sec:marketswithtransferablepayoff}). 
\subsubsection{Pure exchange economies}\label{sec:pureexchangeconomy}
There are no production opportunities in a pure exchange economy (or, simply, an exchange economy), and agents trade initial stocks, or endowments, of goods (or commodities) that they possess according to a specific rule and attempt to maximize their preferences or utilities. Generally, an exchange economy consists of a list $\Omega = (N, l, (w_i),  (u_i))$, where: 
\begin{description}
\item (a) $N$ is a finite set of agents ($|N|=n<\infty$); 
\item (b) $l$ is a positive integer (the number of goods or commodities);
\item (c) the vector $w_i $ is agent $i$'s endowment vector ($w_i \in X_{i} \subseteq \mathbb{R}_{+}^{l}$), with $\mathbb{R}_{+}$ being the set of non-negative real numbers, and $X_{i}$ the agent $i$'s consumption set; and 
\item (d) $u_i: X_i \longrightarrow \mathbb{R}$ is agent  $i$'s utility function.
\end{description}
The amount of good $k$ that agent $i$ demands in the market is denoted $x_{ik}$, so that agent $i$'s consumption bundle is denoted $x_i=(x_{11}, x_{12},..., x_{1l})\in X_i$. An \textit{allocation} is a distribution of the total endowment among agents: that is, an outcome $x=(x_j)_{j\in N}$, with $x_j \in X_j$ for all $j\in N$ and $\sum\limits_{j\in N} x_j \leq \sum\limits_{j\in N} w_j$. A \textit{competitive equilibrium} of an exchange economy is a pair $(p^{*}, z^{*})$ consisting of a vector $p^{*} \in \mathbb{R}^{l}_{+}$, with $p^{*} \neq 0$ (the price vector), and an allocation $x^{*}=(x^{*}_j)_{j\in N}$ such that, for each agent $i$, we have: 
 \begin{equation*}
 p^{*} x_i^{*} \leq p^{*} w_i, \  \text{and} \ u_i(x_i^{*}) \geq u_i(x_i) \ \text{for which} \ p^{*} x_i \leq p^{*} w_i, \ x_i \in X_i.    
 \end{equation*}
We say that $x^{*}=(x^{*}_j)_{j\in N}$ is a \textit{competitive allocation}. 

In an exchange economy, we can assimilate an agent's consumption bundle to that agent's action in the market. In that respect, we can formulate an exchange economy under mild assumptions as a free and fair economy. Consider an exchange economy $\Omega = (N, l, (w_i), (u_i))$ in which the number of goods is finite ($l < \infty$), and each agent $i$'s consumption set $X_i$ is finite ($|X_i| < \infty$). For instance, one can assume that agents can only purchase or sell indivisible units of goods in the market. We can model $\Omega$ as a free and fair economy $\mathcal{E}^{\Omega}= (N, X=\times_{j\in N} X_j, o, F, \vect{Sh}, \overline{u})$ where:
\begin{description}
\item (i)  each agent $i$'s action $x_i \in X_i$;
\item (ii) the reference outcome $o$ is the vector of endowments $w$;
\item (iii) $F:  X \longrightarrow \mathbb{R}$ is the net aggregate utility function, i.e., for $x= (x_j)_{j\in N} \in X$, 
\begin{equation*}
F(x) = \sum\limits_{j\in N} [u_j(x_{j}) - u_j(w_j)], \ \text{with} \ F(w)=0;\ \text{and}\   
\end{equation*}
\item (iv) the Shapley allocation scheme $\vect{Sh}= \overline{u}$ distributes the net aggregate utility $F(x)$ between agents at each profile $x \in X$: $\overline{u}_i(x) = \vect{Sh}_i(F, x)$ for each $i\in N$. 
\end{description}
Only allocations in the free and fair economy can be selected in the equilibrium. This means that an outcome $x = (x_j)_{j\in N} \in X$ is an equilibrium in the free and fair economy if
\begin{description}
\item (1) $\sum\limits_{j\in N} x_{j} \leq \sum\limits_{j\in N} w_j$, and
\item (2) $x$ is a pure strategy Nash equilibrium of the strategic game $(N, X, \vect{Sh})$.
\end{description}
Our model differs from the exchange economy in at least two important respects. First, the incentive mechanism is different. Second, the equilibrium prediction from free exchanges between agents in both economies is different in general.  A competitive equilibrium exists in an exchange economy when some assumptions exist on agents' utilities and endowments. For instance, when utilities are continuous, strictly increasing, and quasi-concave and each agent initially owns a positive amount of each good in the market, a competitive equilibrium exists, and many equilibria might arise. However, under such assumptions on agents' utilities, the net aggregate utility function $F$ is strictly increasing, and thanks to Theorem \ref{monotonicresult2}, the free and fair economy admits a unique equilibrium. Additionally, it is not necessary to impose any assumptions on utilities and endowments to guarantee the existence of an equilibrium in a free and fair economy. We illustrate these points in the following examples.

\begin{example}
Consider an exchange economy with two goods (1 and 2) and two agents (A and B) in which agent A initially owns a positive amount of good 1, $w_A=(1, 0)$, while agent B owns a positive amount of both goods, $w_B=(2, 1)$. We assume that agent A's consumption set is $X_A=\{(1,0), (0, 0)\}$ and utility is $u_A(x_A)= u_A(x_{A1}, x_{A2})= x_{A1} + x_{A2}$. Agent B's consumption set is $X_B=\{(2,1), (1,1), (0, 1), (2,0), (1,0), (0,0)\}$ and utility is $u_B(x_B) = u_B(x_{B1}, x_{B2})= \min \{x_{B1}, x_{B2}\}$. An allocation $x=(x_A, x_B)\in X_A \times X_B$ is such that $x_{A1}+x_{B1}\leq 3$ and $x_{A2}+x_{B2}\leq 1$. We can show that there is no competitive equilibrium in this exchange economy (one reason is the fact that agent A owns zero units of good 2), while the free and fair economy admits two equilibria $x_1^{\vect{Sh}} = (w_A, w_B)$ and $x_2^{\vect{Sh}} = (w_A, (1,1))$. Each equilibrium maximizes the net aggregate utility, $F(x_1^{\vect{Sh}})=F(x_2^{\vect{Sh}})=0$, with $\vect{Sh}_A(F, x_1^{\vect{Sh}})=\vect{Sh}_B(F, x_1^{\vect{Sh}})=0$, and $\vect{Sh}_A(F, x_2^{\vect{Sh}})=\vect{Sh}_B(F, x_2^{\vect{Sh}})=0$. This example shows that a free and fair exchange economy has an equilibrium while a competitive equilibrium does not exist. The next example will show that the equilibrium of a free and fair exchange economy can coincide with the competitive equilibrium.
\end{example}


\begin{example}\label{pureexchange}
Consider a Shapley-Shubik economy \citep{shapley1977example} in which there are two agents and two goods. Agent A is endowed with 2 units of good 1, $w_A=(2,0)$, and agent B is endowed with 2 units of good 2, $w_B=(0, 2)$. We assume that agent A's consumption set is $X_A=\{(0,0), (0, 1), (0, 2), (1, 0), (1,1), (1,2), (2,0), (2,1), (2,2)\}$ and his or her utility function is $u_A(x_{A1}, x_{A2}) = x_{A1} + 3x_{A2} - \frac{1}{2} (x_{A2})^2$; agent B's consumption set is $X_B=\{(0,0), (1,0), (2,0), (0,1), (1,1), (2,1), (0,2), (1,2), (2,2)\}$ and his or her utility function is $u_B(z_{B1}, x_{B2}) = x_{B2} + 3x_{B1} - \frac{1}{2} (x_{B1})^2$. Assume that good 1 is the numeraire ($p_1=1$), and let $p=p_2$ and $X=X_A \times X_B$. It is straightforward to note that not all pairs of actions in $X$ are feasible in the economy. We can show that the pair $E^{*}= (p^{*}, x^{*})$, where $p^{*}=1$, and $x^{*}= (x^{*}_A= (0, 2), x^{*}_B= (2, 0))$, is the unique competitive equilibrium of the market. At the equilibrium allocation $(p^{*}, x^{*})$, agents exchange endowments, and that transaction results in utilities: $u_A(x_1^{*})= u_B(x_2^{*})=4$. Similarly, strategic interactions among agents in the free and fair market yield the same outcome $x^{*}$. To show that result, we use an approach that allows us to simplify calculations in the free and fair economy. 

Let us denote by $\overline{X}$ the subset of allocations ($\overline{X} \subset X$), and consider the following decisions: $a$ ``keep the full endowment", $b$ ``sell 1 unit of good", and $c$ ``sell the full endowment." Consider $\overline{X}_A=\overline{X}_B=\{a, b, c\}$ as each agent's set of decisions.  Each vector of decisions in $\overline{X}_A \times \overline{X}_B$ yields a unique outcome $(x_{A}, x_{B}) \in \overline{X}$. Precisely, the vector $(a, a)$ entails the unique profile $x=(w_A, w_B)=((2, 0), (0, 2))$; $(a, b)$ corresponds to $x=((2, 1), (0, 1))$;  $(a, c)$ corresponds to $x=((2, 2), (0, 0))$; $(b, a)$ corresponds to $x=((1,0), (1, 2))$; $(b, b)$ corresponds to $x=((1, 1), (1, 1))$; $(b, c)$ corresponds to $x=((1, 2), (1, 0))$; $(c, a)$ corresponds to $x=((0, 0), (2, 2))$; $(c, b)$ corresponds to $x=((0, 1), (2, 1))$; and $(c, c)$ corresponds to $x=((0, 2), (2, 0))$. The net aggregate utility function $F$ is defined as: $F(x)= F(x_A, x_B)= u_A (x_{A}) + u_B(x_{B})-4$. Using the strategy profile $(a, a)$ as the reference point, Table \ref{markerallocation1} describes agents' utilities in the free and fair economy. For each agent, decision $c$ strictly dominates decisions $a$ and $b$. It follows that the vector $(c, c)$ which corresponds to the outcome $x^{\vect{Sh}}=((0, 2), (2, 0)) = x^{*}$ is  the unique equilibrium in the free and fair economy. In this case, the equilibrium coincides with the competitive allocation.
\begin{table}[htb]\hspace*{\fill}%
\begin{game}{3}{3}[Agent~A][Agent~B]
& $a$ & $b$ & $c$ \\
$a$ &$(0, 0)$ &$(0, 1.5)$ &$(0, 2)$ \\
$b$ &$(1.5, 0)$ &$(1.5, 1.5)$ &$(1.5, 2)$ \\
$c$ &$(2, 0)$ &$(2, 1.5)$ &$(2, 2)$ 
\end{game}\hspace*{\fill}\\
\caption{Utilities in the free and fair economy}\label{markerallocation1}
\end{table}
\end{example}

\subsubsection{Markets with transferable payoff}\label{sec:marketswithtransferablepayoff}

\ \ \ A market with transferable payoff is a variant of a pure exchange economy in which each agent in the economy is endowed with a bundle of goods that can be used as inputs in a production system that the agent operates. All production systems transform inputs into the same kind of output (i.e., money), and this output can be transferred between the agents. In a market, the payoff can be directly transferred between agents, while in a pure exchange economy only goods can be directly transferred. Following \citet{osborne1994course}, a market with transferable payoff consists of a list $\Pi = (N, l, (w_i), (f_i), (u_i))$, where: 
\begin{description}
\item (a) $N$ is a finite set of agents ($|N|=n<\infty$); 
\item (b) $l$ is a positive integer (the number of input goods);
\item (c) the vector $w_i$ is agent $i$'s endowment vector ($w_i \in X_i \subseteq \mathbb{R}_{+}^l$), with $X_{i}$ being the agent $i$'s input set; 
\item (d) $f_i: X_i \longrightarrow \mathbb{R}$ is agent  $i$'s continuous, non-decreasing, and concave production function; and
\item (e) $u_i$ is agent $i$'s utility function: $u_i(f_i, p, x_i)= f_i(x_i) - p(x_i-w_i)$, with $p \in \mathbb{R}^{l}_{+}$ (the vector of positive input prices), and $x_i \in X_i$.
\end{description}
In the market, an \textit{input vector} is a member of $X_i$, and a profile $(x_j)_{j\in N}$ of input vectors for which $\sum\limits_{j\in N} x_j \leq \sum\limits_{j\in N} w_j$ is an \textit{allocation}. We denote $w= (w_j)_{j\in N}$. Agents can exchange inputs at fixed prices $p \in \mathbb{R}^{l}_{+}$, which are expressed in terms of units of output. At the end of the trade, if agent $i$ holds the bundle $x_i$, then his or her net expenditure, in units of output, is $p (x_i - w_i)$. Agent $i$ can produce $f_i (x_i)$ units of output, so that his or her net utility is $u_i(f_i, p, x_i)$. A price vector $p^{*} \in \mathbb{R}^{l}_{+}$ generates a \textit{competitive equilibrium} if, when agent $i$ chooses his or her trade to maximize his or her utility, the resulting profile $(x^{*}_i)_{i\in N}$ of input vectors is an allocation. Formally, a \textit{competitive equilibrium} of a market is a pair $(p^{*}, (x^{*}_i)_{i\in N})$ consisting of a vector $p^{*} \in \mathbb{R}^{l}_{+}$ and an allocation $(x^{*}_i)_{i\in N}$ such that, for each agent $i$, the vector $x^{*}_i$ maximizes his or her utility $u_i(f_i, p^{*}, x_i)$, for each $x_i \in X_i$. The list $(N, l, w, (f_i), (u_i))$ defines a competitive market with transferable payoff.  

\vspace{1mm}
In a market with transferable payoff, we can view an agent's input vector as an agent's action in the market.
Therefore, as in section \ref{sec:pureexchangeconomy}, we can write a market with transferable payoff under mild assumptions as a free and fair economy. Consider a market with transferable payoff $\Pi = (N, l, w, (f_i), (u_i))$ in which the number of input goods is finite ($l < \infty$), and each agent $i$'s input set $X_i$ is finite ($|X_i| < \infty$). As in Section \ref{sec:pureexchangeconomy}, we can model $\Pi$ as a free and fair market $\mathcal{E}^{\Pi} = (N, X=\times_{j\in N} X_j, o, F, \vect{Sh}, \overline{u})$, with the difference that for $x= (x_j)_{j\in N} \in X$, 
\begin{equation*}
F(x) = \sum\limits_{j\in N} [f_j(x_{j}) - f_j(w_j)].   
\end{equation*}

As in the analysis in section \ref{sec:pureexchangeconomy} below, we provide examples that show similarities (Example \ref{marketwithtransferEx1}) and differences (Example \ref{marketwithtransferEx2}) between the predictions of free and fair markets and markets with transferable payoff.

\vspace{2mm}

\begin{example}\label{marketwithtransferEx1}
We consider a single-input market with transferable payoff in which there are two homogeneous agents who have the same production, $w_1 = w_2=1$, $f_i(x_i) = \sqrt{x_i}$, $i\in \{1,2\}$, and $X_1=X_2=\{0, 1, 2\}$. The pair $E^{*}= (p^{*}=\frac{1}{2}, x^{*}=(w_1, w_2))$ is the unique competitive equilibrium of the market, and $u_1(p^{*}, x_1^{*})= u_2(p^{*}, x_2^{*})=1$. Similarly, strategic interactions among agents in the free and fair market yield the same outcome $x^{*}$. 
\end{example}

\begin{example}\label{marketwithtransferEx2}
As mentioned in Section \ref{sec:pureexchangeconomy}, generally, the equilibrium predictions of a free and fair economy and a market with transferable payoff do not coincide. To showcase this point, we consider a market in which agents' production functions are not concave. Consider a single-input market with transferable payoff in which there are two heterogeneous agents in production: $w_1= 1$, $X_1=\{0, 1,2,3\}$, and $f_1(x_1) = \frac{1}{2} x_1^{2}$; and $w_2=2$, $X_2=\{0, 1,2,3\}$, and $f_2(x_2) = x_2^{2}$. In the competitive market, the utility functions are convex and given by: $u_1(p, x_1)=\frac{1}{2} x_1^{2}-p(x_1-1)$ and $u_2(p, x_2)=z_2^{2}-p(x_2-2)$. There is no exchange in this market, while strategic interactions among agents in the  free and fair market yield a different outcome: $x^{\vect{Sh}}=(0, 3)$. 
\end{example}


\section{Contributions to the closely related literature} \label{sec:literature}

\ \ \ In this paper, we propose a model of a free and fair economy, defining a new class of non-cooperative games, and we apply it to a variety of economic environments. We prove that four elementary principles of distributive justice, of long tradition in economic theory, guarantee the existence of a pure strategy Nash equilibrium in finite games. In addition, we show that when an economy violates these principles, a pure strategy equilibrium may not exist, resulting in instability in agents' actions and in income volatility.
We extend this model to incorporate social justice and inclusion. In this more general model, we also prove several results on equilibrium existence and efficiency.

Our work contributes to several literatures. It is related to studies of group incentives in multi-agent problems under certainty. \citet{holmstrom1982moral} explores the effects of moral hazard in individual incentives and efficiency in organizations with and without uncertainty. Like \citet{holmstrom1982moral}, we consider that in a free economy, any agent has the freedom to choose any action (or input) from his or her set of strategies, and the combination of actions from agents generates a measurable output. However, unlike \citet{holmstrom1982moral}, there is no uncertainty in the supply of inputs, and we assume that our allocation scheme follows basic principles of distributive justice. It follows that our scope, analysis and applications are very different. Moreover,  \citet{holmstrom1982moral} finds an impossibility result in his setup (see, \citet[Theorem 1, p. 326]{holmstrom1982moral}), but our analysis implies that this result does not extend under fair principles in a framework with finite action sets. Moreover, we show that any free and fair economy which is strictly monotonic admits a unique equilibrium, and this equilibrium is optimal and Pareto-efficient (Theorem   \ref{monotonicresult2}). Our findings therefore underscore the role of justice in shaping individual incentives, stabilizing contracts among private agents, and enhancing welfare.

By incorporating normative principles into non-cooperative game theory, we have introduced a new class of finite strategic form games that always admit a Nash equilibrium in pure strategies. We view this paper as contributing to the small but growing literature that seeks to uncover conditions under which a pure strategy Nash equilibrium exists in a non-cooperative game with simultaneous moves. \citet{nash1951non} shows a very prolific result on the existence of equilibrium points in a finite non-cooperative games. \citet{nash1951non} also shows that there always exists at least one pure strategy equilibrium in finite symmetric
games. However, \citet{nash1951non} was silent about the existence of pure strategy equilibrium in either finite or infinite non-symmetric strategic form games. Subsequent research has searched for sufficient and necessary conditions for the existence of pure strategy Nash equilibrium in different structure of strategic form games. Early contributions in this respect include, among others,  \citet{debreu1952social}, \citet{glicksberg1952further}, \citet{gale1953theory}, \citet{schmeidler1973equilibrium}, \citet{mas1984theorem}, \citet{khan1995pure}, \citet{athey2001single} in continuous games;  \citet{dasgupta1986existence}, \citet{dasgupta1986existence1}, \citet{reny1999existence}, \citet{carbonell2011existence}, \citet{reny2016nash},  \citet{nessah2016existence} in discontinuous economic games; \citet{monderer1996potential} in potential games; and \citet{ziad1999pure} in fixed-sum games. In these studies, scholars use different concepts of continuity, convexity and appropriate fixed point results along with some restrictions on utility functions to prove the existence of a pure strategy Nash equilibrium. Other contributions that  guarantee the existence of equilibrium in pure strategies for finite games include, among others, \citet{rosenthal1973class}, \citet{mallick2011existence}, \citet{carmona2020pure}, and the references listed therein. We follow a different approach from this literature. Unlike our paper, this literature has not approached the issue of equilibrium existence in a non-cooperative game from a normative angle. We also apply our theory to different economic environments, including applications surplus distribution in a firm, exchange economies, self-enforcing lockdown in networked economies facing contagion, and bias in academic publishing.

Finally, in addition to the previous point, our work can also be viewed as contributing to the Nash Program \citep{nash1953two}, which bridges non-cooperative and cooperative game theory. However, we significantly depart from the main approach taken in this literature so far. This approach has generally sought to define a non-cooperative game whose solution coincides with the outcomes of a cooperative solution concept; see \citet{serranosixty} for a recent survey on this literature. Our approach, on the contrary, follows the opposite direction. It asks if equilibrium can be found in a strategic form game in which payoffs obey natural axioms inspired by cooperative game theory.

\section{Conclusion}\label{conclusion}
\ \ \ In this paper, we examine how elementary principles of justice and ethics, of long tradition in economic theory, affect individual incentives in a competitive environment and determine the existence and efficiency of self-enforcing social contracts. To formalize this problem, we introduce a model of a \textit{free and fair economy}, in which each agent freely and non-cooperatively chooses their input from a finite set, and the surplus generated by these choices is distributed following four ideals of market justice, which are anonymity, local efficiency, unproductivity, and marginality. We show that these ideals guarantee the existence of a pure strategy Nash equilibrium. However, an equilibrium need not be unique or Pareto-efficient. We uncover an intuitive condition---\textit{strict technological monotonicity}---, which guarantees equilibrium uniqueness and efficiency. Interestingly, this condition does not guarantee equilibrium efficiency (or even existence) when ideals of justice are violated in an economy. These ideals therefore lead to positive incentives, given their desirable equilibrium and efficiency properties. 

We extend our analysis to incorporate social justice and inclusion, implemented in the form of progressive taxation and redistribution and guaranteeing a basic income to unproductive agents. In this more general setting, we generalize all of our findings. In addition, we examine how the tax policy affects efficiency, showing that there is a tax rate threshold above which an equilibrium  that is Pareto-efficient always exists in the economy, even in the absence of technological monotonicity. Moreover, we show that if a free economy is able to choose its reference point, it can always do so to induce an efficient outcome that is self-enforcing, even if this economy is not monotonic.

By incorporating normative principles into non-cooperative game theory, we have defined a new class of finite strategic form games that always admit a pure strategy Nash equilibrium. We develop applications to some classical and recent economic problems, including the allocation of goods in an exchange economy, surplus distribution in a firm, self-enforcing lockdown in a networked economy facing contagion, and publication bias in academic publishing. This variety of applications is possible because we impose no particular assumptions on the structure of agents' action sets, and our setting is fully non-parametric.



\newpage
\bibliographystyle{plainnat}
\bibliography{main}

\newpage
\section*{Appendix}

\textbf{Proof of Proposition \ref{uniqueshapley}}.

\textbf{Sufficiency}. We show that the allocation scheme $\vect{Sh}$ satisfies \textbf{ALUM}.

\textbf{Anonymity}. Let $f\in P(X)$, $x\in X$, $\pi^x \in \mathcal{S}^x_{n}$, and $i$ be an agent. 
We show that $\vect{Sh}_{i}(\pi^x f^x, x) = \vect{Sh}_{\pi^x(i)}(f^x, x)$.
\begin{enumerate}
    \item If $i\notin N^x$, then $x_i =o_i$, and $\pi^x(i)=i$. 
\begin{equation*} 
\begin{split}
\vect{Sh}_{i}(f^x, x) &= \sum\limits_{a \in \Delta_{0}^{i}(x)}\varphi (a, x) \left\{f^x(a+x_{i}e_{i})- f^x(a)\right\}\\
&= \sum\limits_{a \in \Delta_{0}^{i}(x)}\varphi (a, x) \left\{f^x(a)- f^x(a)\right\}\\
&=0.
\end{split}
\end{equation*}  
Similarly,
\begin{equation*} 
\begin{split}
\vect{Sh}_{i}(\pi^xf^x, x) &= \sum\limits_{a \in \Delta_{0}^{i}(x)}\varphi (a, x) \left\{\pi^xf^x(a+x_{i}e_{i})- \pi^xf^x(a)\right\}\\
&= \sum\limits_{a \in \Delta_{0}^{i}(x)}\varphi (a, x) \left\{f^x(\pi^x(a+x_ie_i))- f^x(\pi^x(a))\right\}.
\end{split}
\end{equation*} 
 For $a \in \Delta_{0}^{i}(x)$ and $x_i=o_i$, we have $\pi^x(a+x_ie_i)= \pi^x(a)$, and $\vect{Sh}_{i}(\pi^xf^x, x)=0$. Therefore, for each $i\notin N^x$, we can conclude that $\vect{Sh}_{i}(\pi^x f^x, x) = \vect{Sh}_{\pi^x(i)}(f^x, x)$.

\item If $ i\in N^x$, then $x_i \neq o_i$. Assume that $\pi^x(i)=j$. Then, $j\in N^x$ and $x_j \neq o_j$.
\begin{equation*} 
\begin{split}
\vect{Sh}_{j}(f^x, x) &= \sum\limits_{a \in \Delta_{0}^{j}(x)}\varphi (a, x) \left\{f^x(a+x_{j}e_{j})- f^x(a)\right\}\\
&= \sum\limits_{a \in \Delta_{0}^{j}(x)}\varphi (a, x) \left\{f(a+x_j e_j)- f(a)\right\}).
\end{split}
\end{equation*}
Similarly,
\begin{equation*} 
\begin{split}
\vect{Sh}_{i}(\pi^xf^x, x) &= \sum\limits_{a \in \Delta_{0}^{i}(x)}\varphi (a, x) \left\{\pi^xf^x(a+x_{i}e_{i})- \pi^xf^x(a)\right\}\\
&= \sum\limits_{a \in \Delta_{0}^{i}(x)}\varphi (a, x) \left\{f^x(\pi^x(a+x_ie_i))- f^x(\pi^x(a))\right\}\\
&= \sum\limits_{a \in \Delta_{0}^{i}(x)}\varphi (a, x) \left\{f(\pi^x(a+x_ie_i))- f(\pi^x(a))\right\}
\end{split}
\end{equation*} 
$a \in \Delta_{0}^{i}(x)$ implies $a= (a_1, ...,\underbrace{o_i}_{i^{\text{th} \ \text{component}}},...,a_n)$. The vector $\pi^x(a)= (\pi^x_1(a), ...,\underbrace{\pi^x_j(a)}_{j^{\text{th} \ \text{component}}},...,\pi^x_n(a))$. Given that $j= \pi^x(i)$ and $a_i =o_i$, it follows that $\pi^x_j(a)= o_j$ and $\pi^x(a) \in \Delta_{0}^{j}(x)$. We also have $a+x_ie_i= (a_1, ...,\underbrace{x_i}_{i^{\text{th} \ \text{component}}},...,a_n)$. Given that $j= \pi^x(i)$ and $(a+x_ie_i)_i =x_i \neq o_i$, it follows that $\pi^x_j(a+x_ie_i)= x_j$. Note that we can write $\pi^x(a+x_ie_i)= \pi^x(a) + x_je_j$. Therefore,
\begin{equation*} 
\begin{split}
\vect{Sh}_{i}(\pi^xf^x, x) &= \sum\limits_{a \in \Delta_{0}^{i}(x)}\varphi (a, x) \left\{f(\pi^x(a) +x_je_j)- f(\pi^x(a))\right\}\\
&= \sum\limits_{b \in \Delta_{0}^{j}(x)}\varphi (b, x) \left\{f(b +x_je_j)- f(b)\right\}, \ \text{where} \ b=\pi^x(a)\\
&=\vect{Sh}_{j}(f^x, x).
\end{split}
\end{equation*} 
\end{enumerate}
It follows that the allocation $\vect{Sh}$ satisfies $x$-Anonymity for each $x\in X$. Hence, $\vect{Sh}$ satisfies Anonymity.
\vspace{2mm}

\textbf{Local Efficiency.} For any $f \in P(X)$ and $x \in X$, it is immediate that $\sum\limits_{i\in
N}\vect{Sh}_{i}(f, x) = f(x)$.

\vspace{2mm}

\textbf{Unproductivity}. If  agent $i$ is unproductive, then for any $f \in P(X)$ and $x \in X$, it is immediate that $\vect{Sh}_{i}(f, x) = 0$, since $mc(i, f, a, x)=0$ for each $a \in \Delta_0^{i} (x)$.

\textbf{Marginality.} Let $f, g \in P(X)$ such that $mc(i,f, x',x)\geq mc(i,g, x',x)$ for all $i\in N$, $x \in X$ and $x' \in \Delta_o^{i}(x)$. By the definition of the value $\vect{Sh}$, it is immediate that $\vect{Sh}_{i}(f, x)\geq \vect{Sh}_{i}(g, x)$.

\textbf{Necessity.} In this part of the proof, we prove the uniqueness of the Shapley value. Consider another allocation procedure $\phi$ which satisfies \textbf{ALUM}.

\vspace{2mm}

Define the following production function $f_{x} \in P(X)$ for each $x\in X$ by:
\begin{equation*}
{f_{x}(y)=\left\{
\begin{array}{ccc}
1 & if & x \in \Delta (y) \\
0 & if & x \notin \Delta (y) %
\end{array}%
\right. }
\end{equation*}%

where $x \in \Delta (y)$ if and only if $[x_i \neq y_i \Rightarrow x_i = o_i]$. 

\begin{lemma}[\citet{PongouTondji2018}]
\label{app_lemma_productionfunctionbasis} 
Any production function is a linear combination of the production functions $f_{x}$:
\begin{equation*}\label{functionf}
f=\sum\limits_{x\in X}c_{x}(f)f_{x},\ \text{where}\
c_{x}(f)=\sum\limits_{x^{\prime } \in \Delta (x)}(-1)^{|x|-|x^{\prime
}|}f(x^{\prime }).
\end{equation*}
\end{lemma}

Let $f \in P(X)$. Define the index $I$ of the production function $f$ to be the number
of non-zero terms in some expression for $f$ in (\ref{functionf}). The theorem is proved by induction on $I$. 

\begin{description}
\item a) If $I=0$, then $f\equiv 0$. Let $x \in X$ and $i\in N$. Then, $mc(i,f,a,x)=0$ for  all $a \in X$ such  that $a \in \Delta_o^{i}(x)$. Therefore, by Unproductivity, $\vect{Sh}_{i}(f, x)=\phi_i(f, x)=0$. 

\item b) If $I=1$, then $f=c_{x}(f)f_{x}$ for some $x\in X$. Consider $N^{x}=\left\{ l\in N:x_{l}\neq o_l\right\}$.

\textbf{Step 1}. Let $i\notin N_{x}$, i.e., $x_i=o_i$. 

For any $a \in X$ such that $a \in \Delta_{0}^{i} (x)$, we have $f(a+x_{i}e_{i})-f(a)=0$, i.e., $mc(i,f,a,x)=0$. It follows that $\vect{Sh}_i(f, x)=0$. Let $y \in X$ with $y \neq x$. Then, $x \in \Delta (y)$ or $x \notin \Delta (y)$. 

\begin{itemize}
\item If $x \in \Delta (y)$, then $x_l= y_l$ for each $l \in N^x$. If $y_i = o_i$, then $\phi_i(f, y)=0=\vect{Sh}_i(f, y)$. Assume $y_i \neq o_i$. Then, for any $a \in \Delta_{0}^{i} (y)$, we have $mc(i, f, a, y)= f(a+y_ie_i)-f(a)$. If $x \in \Delta (a)$, we also have $x \in \Delta (a+y_ie_i)$ because $x_i=o_i$ and $y_i \neq o_i$. Similarly if $x \notin \Delta (a)$, then $x \notin \Delta (a+y_ie_i)$. Therefore, $mc(i, f, a, y)=0$ for each $a \in \Delta_{0}^{i} (y)$, and $\vect{Sh}_i(f, y)=0$.

\item $x \notin \Delta (y)$, then $f(y)=0$. If $y_i = o_i$, then $\vect{Sh}_i(f, y)=0$. Assume $y_i \neq o_i$. Then, for any $a \in \Delta_{0}^{i} (y)$, we have $mc(i, f, a, y)= f(a+y_ie_i)-f(a)$. If $x \in \Delta (a)$, then for each $l\in N^x$, $x_l=a_l\neq o_l$. Or $a\in \Delta (y)$ implies that for each $l\in N^x$, we will have $a_l=y_l$, because $a_l \neq o_l$. Therefore, for each $l\in N^x$, $a_l=y_l=x_l$, and given that $y_i \neq o_i$ and $x_i=o_i$, we have $x \in \Delta (y)$, a contradiction. In fact $x \in \Delta (a)$ if and only if $x \in \Delta (a+y_ie_i)$. Thus, $mc(i, f, a, y)=0$ for each $a \in \Delta_{0}^{i} (y)$, and $\vect{Sh}_i(f, y)=0$.
\end{itemize}

Given that agent $i$ is unproductive, it follows that $\phi _{i}(f, y)=\vect{Sh}_i(f, y)=0$ for each $y\in X$.

\textbf{Step 2}. Let $i,j\in N$ such that $i,j\in N^{x}$ and $y \in X$. Let $\pi^y = (ij)$ a permutation. Given that $\phi$ satisfies Anonymity, it follows that $\phi$ satisfies $y$-Anonymity, and  $\phi_i(\pi^x f^y, y) = \phi_{j} (f^y , y)$. For each $z \in \Delta (y)$, we have $\pi^yf^y(z) = f^y(z)$. Thus, $\pi^y f^y = f^y$, and $\phi_i(f^y, y) = \phi_{j} (f^y , y)$. By Local efficiency, $\sum \limits_{k\in N^x} \phi_k(f^y, y)= f^y(y)= f(y)$. Therefore, $\sum \limits_{k\in N^x} \phi_k(f^y, y) = |N^x|\phi_k(f^y, y)$, and for each $k\in N^x$, $\phi_k(f^y, y)= \frac{f^y(y)}{|N^x|}=\frac{f(y)}{|N^x|}$. If $x \in \Delta (y)$, then $f(y)= c_x(f)$. Otherwise, $f(y)=0$, and for each $k\in N^x$, $\phi_k(f, y) = \phi_k(f^y, y) = \vect{Sh}_k(f, y)$.

\item c) Assume now that $\phi $ is the value $\vect{Sh}$ whenever the index of $f$
is at most $I$ and let $f$ have index $I+1$, with:
\begin{equation*}
f=\sum\limits_{k=1}^{I+1}c_{x^{k}}(f)f_{x^{k}},\ \text{all}\ c_{x^{k}}\neq 0,\ \text{and}\
x^{k}\in X.
\end{equation*}%
For $k\in \left\{ 1,2,...,I+1\right\} $, consider:
\begin{equation*}
N^{x^k}=\left\{ l\in N:x_{l}^{k}\neq o_k\right\} ,\
\overline{N}=\bigcap\limits_{k=1}^{I+1}N^{x^k},\ \text{and assume} \ i\notin \overline{N}.
\end{equation*}%
Define the following production function:
\begin{equation*}
g=\sum\limits_{k:i\in N^{x^k}}c_{x^{k}}(f)f_{x^{k}}.
\end{equation*}%
The index of $g$ is at most $I$. Let $x,a\in X$ such that $a \in \Delta_{0}^{i}(x)$.
Then $f(a+x_{i}e_{i})-f(a)=g(a+x_{i}e_{i})-g(a)$. Consequently, using Marginality, $\phi _{i}(f,x)=\phi _{i}(g, x)$. By induction, we have:
\begin{equation*}
\phi _{i}(f, x)=\sum\limits_{k: i\in N^{x^k}} \frac{c_{x^{k}}(f) f_{x^{k}} (x) }{|x^{k}|}=\vect{Sh}_i(f, x), \ \text{for} \ x\in X.
\end{equation*}%

It remains to show that for each $x \in X$, $\phi _{i}(f,x)=\vect{Sh}{i}(f,x)$ when $i\in \overline{N}$. Let $x \in X$. By Anonymity, $\phi _{i}(f, x)$ is a constant $\varphi$ for all members of $\overline{N}$; likewise the value $\vect{Sh}_{i}(f, x)$ is some constant $\varphi'$ for all members of $\overline{N}$ (with $\overline{N}>0)$. By Local efficiency, 
\begin{equation*}
|\overline{N}|\phi _{i}(f, x)= |\overline{N}| \varphi= f(x),
\end{equation*}%
so that,
\begin{equation*}
\varphi=\frac{f(x)}{|\overline{N}|}.
\end{equation*}
Similarly, 
\begin{equation*}
|\overline{N}|\vect{Sh} _{i}(f, x)= |\overline{N}| \varphi'= f(x),
\end{equation*}%
so that,
\begin{equation*}
\varphi'=\frac{f(x)}{|\overline{N}|}.
\end{equation*}
It follows that $\varphi =\varphi'$, and concludes the proof.

\end{description}

\end{document}